\documentclass[11pt]{article}%
\usepackage{tablefootnote}
\usepackage{amsfonts}
\usepackage{fancyhdr}
\usepackage{comment}
\usepackage[a4paper, top=2.5cm, bottom=2.5cm, left=2cm, right=2cm]%
{geometry}
\usepackage{amsmath,amssymb,amsfonts}
\usepackage{graphicx}
\usepackage{float}
\usepackage{hyperref}
\usepackage{float}
\usepackage{algorithm}
\usepackage{algorithmic}
\usepackage{xcolor}
\usepackage{enumitem} 
\usepackage{booktabs}
\usepackage{multirow}
\usepackage{stackengine}
\usepackage{verbatim}

\newcommand\xrowht[2][0]{\addstackgap[.5\dimexpr#2\relax]{\vphantom{#1}}}

\newtheorem{theorem}{Theorem}

\newtheorem{claim}[theorem]{Claim}

\newtheorem{corollary}[theorem]{Corollary}

\newtheorem{definition}[theorem]{Definition}

\newtheorem{lemma}[theorem]{Lemma}

\newenvironment{proof}[1][Proof]{\textbf{#1.} }{\ \rule{0.5em}{0.5em}}
\newcommand{\MovingKnife}{\mathsf{MovingKnife}}
\newcommand{\Matching}{\mathsf{Matching}}
\newcommand{\MMS}{\mathsf{MMS}}
\newcommand{\OPT}{\mathsf{OPT}}
\newcommand{\PoF}{\mathsf{PoF}}
\newcommand{\cA}{\mathbf{A}}

\newcommand{\cS}{\mathcal{S}}
\newcommand{\cI}{\mathcal{I}}

\newcommand{\EW}{\mathsf{EW}}
\newcommand{\UW}{\mathsf{UW}}
\newcommand{\calpha}{\boldsymbol\alpha}

\begin{document}

\title{On the Price of Fairness of Allocating Contiguous Blocks}

\author{Ankang Sun$^1$ \hspace{30pt} Bo Li$^2$\\
$^1$Warwick Business School, University of Warwick, United Kingdom\\
\texttt{A.Sun.2@warwick.ac.uk}\\
$^2$ Department of Computing, The Hong Kong Polytechnic University, Hong Kong\\
\texttt{comp-bo.li@polyu.edu.hk}
}
\date{}
\maketitle

\begin{abstract}
In this work, we study the problem of fairly allocating a number of indivisible items that are located on a line to multiple agents. 
A feasible allocation requires that the allocated items to each agent are connected on the line.
The items can be goods on which agents have non-negative utilities, or chores on which the utilities are non-positive.
Our objective is to understand the extent to which welfare is inevitably sacrificed by enforcing the allocations to be fair, i.e., price of fairness (PoF).
We study both egalitarian and utilitarian welfare.
Previous works by Suksompong [Discret. Appl. Math., 2019] and H{\"{o}}hne and van Stee [Inf. Comput., 2021] have studied PoF regarding the notions of envy-freeness and proportionality. 
However, these fair allocations barely exist for indivisible items, and thus in this work, we focus on the relaxations of maximin share fairness and proportionality up to one item, which are guaranteed to be satisfiable. 
For most settings, we give (almost) tight ratios of PoF and all the upper bounds are proved by designing polynomial time algorithms. 
\end{abstract}

\section{Introduction}

\begin{center}
{\em Inequality, not scarcity, that persecutes governors; anarchy, not poverty, that haunts them.}

\medskip

\rightline{-- The Analects of Confucius}
\end{center}

Allocating a set of indivisible items among multiple agents is a fundamental problem in the fields of multi-agent systems and computational social choice \cite{DBLP:journals/corr/abs-2202-07551,DBLP:journals/corr/abs-2202-08713}.
Traditionally, efficiency was at the center but fairness has attracted increasingly more attention in the recent decade.
Two types of fairness notions are widely studied, namely, envy-freeness (EF) \cite{foley1966resource} and proportionality (PROP)  \cite{steinhaus1948problem}.
Intuitively, EF requires that every agent's value for her own bundle is no smaller than that of any other agent's bundle, and PROP requires every agent's value is at least the average value if she is assigned all the items.
When the items are not divisible, exact fairness becomes hard to achieve.
Thus, various ways to relax these fairness notions are proposed, among which ``up to one'' is one of the most widely studied relaxations which requires the requirements to be satisfied after the removal of some items.
Accordingly we have relaxed notions of envy-free up to one item (EF1) and proportional up to one item (PROP1) for indivisible items.
Fortunately, for both goods and chores, an EF1 or a PROP1 allocation always exists and can be found efficiently when the valuations are additive. 
Besides PROP1, maximin share fairness (MMS) \cite{journals/bqgt/Budish10} is also a popular relaxation for PROP which requires each agent's value to be at least her max-min value if she is to allocate the items. 
Although in general, an MMS allocation is not guaranteed for either goods or chores \cite{conf/aaai/AzizRSW17,journals/jacm/KurokawaPW18}, constant approximation algorithms have been designed in the literature \cite{journals/ai/GargT21,conf/sigecom/HuangL21}. 

In this work, we care both utilitarian and egalitarian welfare. 
The utilitarian welfare is the sum of all agents' values, and the egalitarian welfare of an allocation  is the value of the least-off agent. 
Informally, the price of fairness (PoF) regarding criterion X (or price of X for short) is the ratio of the largest economic welfare attainable by any allocation to the greatest economic welfare attained by an X allocation.
The study of PoF is arguably initiated by Caragiannis et al. \cite{DBLP:journals/mst/CaragiannisKKK12} that studied the notions of EF and PROP under both indivisible goods and chores settings and concerned both utilitarian and egalitarian welfare. 
Follow-up works focus more on the utilitarian welfare.
Noting that EF and PROP allocations are not guaranteed to exist, Bei et al. \cite{DBLP:journals/mst/BeiLMS21} extended the research to the fairness criteria whose existence is guaranteed such as EF1.
Recently, Barman et al. \cite{DBLP:conf/wine/BarmanB020} proved that for indivisible goods, the tight ratio of PoF regarding both EF1 and $\frac{1}{2}$-approximate MMS is $\Theta(\sqrt{n})$, where $n$ is the number of agents. 
Unfortunately, when the items are chores, Sun et al. \cite{DBLP:conf/atal/SunCD21} proved that the price of EF1 is infinite.


Note that in the above line of research (e.g., \cite{DBLP:conf/wine/BarmanB020,DBLP:journals/mst/BeiLMS21,DBLP:journals/mst/CaragiannisKKK12,DBLP:conf/atal/SunCD21}), the allocations do not have feasibility constraints. 
In this work, we are particularly interested in the setting when the items are distributed on a line, and a feasible allocation requires that the allocated items to each agent forms a connected block.
This setting is motivated by the cake cutting literature \cite{edward1999rental} where a divisible cake, denoted by interval $[0,1]$, is to be assigned to the agents and each agent is expected to obtain a connected piece.
Many practical applications fit in this model, such as allocating offices to research groups on the same floor.
Suksompong \cite{suksompongFairlyAllocatingContiguous2019} and H{\"{o}}hne and van Stee \cite{hohneAllocatingContiguousBlocks2021} respectively studied the PoF for allocating indivisible goods and chores under connectivity constraints with fairness notions of EF and PROP.
However, since EF and PROP cannot be guaranteed to exist, we aim at investigating the tight bounds of PoF regarding the notions that can be satisfied. 
The good news is when agents only accept connected blocks, an MMS and PROP1 allocation always exists, but the existence of EF1 allocations is still unknown and seems to be a hard problem \cite{BILO2022197,DBLP:conf/ijcai/BouveretCEIP17}.  
Therefore, we focus on MMS and PROP1 allocations and our task is summarized as follows.
\begin{quote}
\em 
    In this work, we explore the (tight) ratios of the price of fairness when the fairness requirement is either MMS or PROP1, the economic welfare is either utilitarian or egalitarian, and the items are either goods or chores that are located on a line.
\end{quote}



\subsection{Main results}
Our main results are summarized in Table \ref{table:main-result:general}.
Note that all the upper bound results are proved by proposed polynomial time algorithms that return fair allocations with the desired efficiency guarantees.
For the ease of comparison, we also provide the known PoF results regarding EF and PROP \cite{hohneAllocatingContiguousBlocks2021,suksompongFairlyAllocatingContiguous2019}.
When the fairness notions are changed from EF and PROP to MMS or PROP1, the ratios of PoF change accordingly, and analyzing the ratios of PoF becomes challenging. 
For example, for the problem of allocating goods, given any PROP allocation (if any), it is straightforward that if the valuations are normalized, each agent has value at least $\frac{1}{n}$ where $n$ is the number of agents, which is already sufficient to achieve the upper bound $\Theta(n)$.
However, this proof technique does not carry over to the notions of MMS and PROP1.
To find the MMS and PROP1 allocations with the best possible efficiency guarantee, we propose two parametric subroutines that is able to balance the individual agent's value and the social welfare.
The first one is to use parameters to identify large items that are allocated to specific agents via a matching where each agent receives at most one item. 
The second subroutine is motivated by the moving knife algorithm where we stand at the very left and collect items until the first (or last) time there exist agents who are satisfied regarding the carefully designed parameters.
By carefully integrating the two subroutines into the algorithm design and properly choosing the parameters, we obtain the upper bound results as shown in Table \ref{table:main-result:general}.
The lower bound results are proved by identifying hard instances where the 
welfare is inevitably sacrificed by enforcing the allocations to be MMS or PROP1.





\begin{table} 
\centering
\begin{tabular}{c|c|c|c|c|c}
\toprule \xrowht[()]{8pt}
General $n$ & EF & PROP  & MMS & PROP1 &\\
\hline \hline \xrowht[()]{9pt}
\multirow{2}{*}{Goods} & 
$\Theta(\sqrt{n})$
& 
$\Theta(n)$
& 
\multicolumn{2}{c|}{LB: $\Omega(\sqrt{n})$ UP: $O(n)$ (Theorems \ref{thm::MMS-good-UW-n>3} and \ref{thm::PROP1-good-UW-n>3})}
& Utilitarian \\
\cline{2-5} \xrowht[()]{9pt}
& $\frac{n}{2}$ & 1 & $\Theta(n)$ (Theorem \ref{thm::MMS-good-RW-n>3}) & $\infty$ (Theorem \ref{thm::PROP1-good-RW-n>3})  &  Egalitarian \\
\hline\hline \xrowht[()]{9pt}
\multirow{3}{*}{Chores} & $\infty$ & $n$ & 
\multicolumn{2}{c|}{$\Theta(n)$ (Theorems \ref{thm::MMS-chore-UW} and \ref{thm::PROP1-chore-UW-n>3})}
& Utilitarian \\
\cline{2-5} \xrowht[()]{9pt}
& \multirow{2}{*}{$\infty$} & \multirow{2}{*}{$1$} & \multirow{2}{*}{$\frac{n}{2}$ (Theorem \ref{thm::MMS-chore-RW})} & $\frac{n}{2}$ for $n\neq 3$;  $2$ for $n=3$  & \multirow{2}{*}{Egalitarian}\\
&&&& (Theorem \ref{thm::PROP1-chore-RW-n>3}) &\\
\bottomrule
\end{tabular}
\caption{Tight ratios of the price of fairness regarding EF, PROP, MMS, PROP1. The ratios for EF and PROP are proved in \cite{suksompongFairlyAllocatingContiguous2019} for goods and in \cite{hohneAllocatingContiguousBlocks2021} for chores. 
}
\label{table:main-result:general}
\end{table}


Following the convention of fair division literature, we are also interested in a typical case when $n = 2$.
For both goods and chores, we show that there always exists an allocation that is simultaneously MMS and PROP1, and EF or PROP if the instance admits such an allocation. 
Interestingly, this allocation achieves the {\em tight} ratio of both utilitarian and egalitarian welfare, as shown in Table \ref{table:main-result:n=2}, which are also identical to the results in \cite{suksompongFairlyAllocatingContiguous2019} and \cite{hohneAllocatingContiguousBlocks2021}.
Particularly, regarding egalitarian welfare, for both goods and chores, and both MMS and PROP1,  the PoF is always $1$, which means the optimal egalitarian welfare can be achieved by MMS and PROP1 allocations.
For utilitarian welfare,  the tight ratio of PoF regarding both MMS and PROP1 is $\frac{3}{2}$ when the items are goods, and that is $2$ when the items are chores.

\begin{table}[H]
\centering
\begin{tabular}{c|c|c|c|c}
\toprule \xrowht[()]{8pt}
\multirow{2}{*}{ $n=2$ } & \multicolumn{2}{c|}{Goods (Theorem \ref{thm:n=2:goods})}  & \multicolumn{2}{c}{Chores Theorem \ref{thm:n=2:chores})}  \\
\cline{2-5} \xrowht[()]{9pt}
& Utilitarian & Egalitarian & Utilitarian & Egalitarian\\
\hline \hline  \xrowht[()]{9pt}
EF/PROP/MMS/PROP1 & $\frac{3}{2}$ &  $1$ & $2$ & 1 \\
\bottomrule
\end{tabular}
\caption{Tight ratios of the price of fairness regarding EF, PROP, PROP1, MMS. 
The ratios for EF and PROP are proved in \cite{suksompongFairlyAllocatingContiguous2019} for goods and in \cite{hohneAllocatingContiguousBlocks2021} for chores. 
}
\label{table:main-result:n=2}
\end{table}

\subsection{Other Related Works}

The traditional research on fair division is centered around allocating a divisible resource, denoted by the real interval $[0,1]$, among a set of heterogeneous agents, i.e., cake cutting problem. 
It is well known that an envy-free (and thus proportional) allocation always exists \cite{brams1995envy} and can be found in finite steps \cite{DBLP:conf/focs/AzizM16}. 
Edward Su \cite{edward1999rental} considered the case where it is required that every agent receives a connected piece.
However, when the items become indivisible, it is a different story.
On one hand, people want to understand how to achieve  approximate envy-freeness and proportionality for both unconstrained settings and the settings when there are restrictions on the allocations \cite{journals/jacm/KurokawaPW18,DBLP:conf/sigecom/LiptonMMS04,DBLP:journals/siamdm/PlautR20}. 
On the other hand, a line of research focuses on investigating how to achieve maximal efficiency while ensuring fairness \cite{DBLP:journals/mst/BeiLMS21,DBLP:journals/mst/CaragiannisKKK12}. 
Besides price of fairness, the compatibility between Pareto optimality and fairness is also widely studied \cite{DBLP:journals/tcs/AmanatidisBFHV21,DBLP:conf/sigecom/BarmanKV18,DBLP:journals/teco/CaragiannisKMPS19}. 
In the recent couple of years, motivated by the real-word applications, constraints have been investigated in company with fair division; we refer the readers to a recent survey by Suksompong \cite{DBLP:journals/sigecom/Suksompong21}.
One of the most frequently studied constraints is {\em connectivity}, where the items are assumed to be distributed on a graph \cite{DBLP:conf/ijcai/BouveretCEIP17,caragiannis2022little} and each agent should receive a connected subgraph.
Line structure as we considered in this work is one typical case. 
Some other constraints that have been considered include matroid~\cite{conf/aaai/BiswasB19,conf/aaai/DrorFS21}, cardinality~\cite{conf/ijcai/BiswasB18,journals/corr/abs-2106-07300}, budget-feasible \cite{conf/ijcai/00010G21}, conflicting~\cite{journals/aamas/HummelH22}, interval scheduling~\cite{li2021fair} and more.

\section{Preliminaries}
\label{sec::pre}
Denote by $[k] = 1,\ldots, k$ for any positive integer $k$. 
A fair division instance $\cI = \langle N, E, \mathcal{V} \rangle$ is composed of $n$ agents $N=\{1,\ldots,n\}$ and $m$ indivisible items $E= \{ e_1, \ldots, e_m \}$. 
The items are placed on a path in the order $e_1,\ldots,e_m$ from the left to the right.
For simplicity, denote by $L (k) = L(e_k) = \{ e _ 1, \ldots, e _ k \} $ the items on the left of $e_k$ including $e_k$ and by $ R ( k ) = R(e_k)= \{ e _ { k}, \ldots, e _  m \}$ the items on the right of $e_k$ including $e_k$.
A feasible allocation allocates each agent a connected bundle of items.
Let $\mathcal{S}$ be the set of all connected bundles. 
Each agent $i$ is associated with a valuation function $ v_i : \mathcal{S} \rightarrow \mathbb{R}$ and $\mathcal{V} = \{v_i\}_{i\in N}$. 
Given an item $e$, we say that $e$ is a {\em good} if for every $ i \in N$, $ v_i (e) \geq 0 $ and $e$ is a {\em chore} if for every $ i \in N$, $ v_i (e) \leq 0 $. 
We consider the case where the all items are either goods or chores, and call the instance {\em fair-goods} or {\em fair-chores} respectively.
The valuations are additive, i.e., $ v _ i ( S ) = \sum _ { e \in S} v _ i ( e )$ for any $ S \in \mathcal{S}$. 
For simplicity, we use $ v _ i ( e _ j )$ to represent $ v _ i ( \{e _ j\} )$.
Without loss of generality, it is assumed that the valuations are normalized, i.e., for all $i \in N$, $v_i ( \emptyset) = 0$, and $v_i (E) = 1$ for goods and  $v_i (E) = -1$ for chores.


A feasible allocation $\mathbf{A} = ( A _ 1, \ldots, A _ n )$ is an $n$-partition of $E$ where every bundle is connected, 
i.e., $ A _ i \cap  A _ j =\emptyset$ for any $ i \neq j$, $ \cup _ { i \in N} A _ i  = E $ and $A_i \in \mathcal{S}$ for each $i \in N$. 
If not explicitly stated otherwise, all the allocations in this paper are assumed to be feasible. 
For any bundle $S$ and any positive integer $ k$, denote by $\Pi _ { k } ( S )$ the set of all $k$-connected partitions of $S$, and by $|S|$ the number of items in $S$. 
Given an allocation $\mathbf{A}$, the {\em utilitarian welfare} (UW) of $\mathbf{A}$ is $\UW(\mathbf{A}) = \sum _ { i \in N} v _ i ( A _ i )$, 
and the {\em egalitarian welfare} (EW) of $\mathbf{A}$ is  $\EW(\mathbf{A}) = \min _ { i \in N} v _ i ( A _ i )$.


\subsection{Fairness Notions}

We next introduce the solution concepts.
Note that although the original definitions do not have any constraints,
to adapt to our setting, all allocations are required to be connected.

An allocation is called {\em proportional} if every agent gets value at least the average of her value for all items. 
Formally, an allocation $\mathbf{A} = ( A _ 1, \ldots, A _ n )$ of a fair-goods instance $\cI$ is {\em proportional} (PROP) if for any $ i \in N$, $ v _ i (A _ i) \geq \frac{1}{n} \cdot v _ i (E)$. 
Note that the requirement of PROP is same for both goods and chores, but the relaxation of PROP1 differs. 

\begin{definition} [PROP1]
	For allocating goods, a connected allocation $\mathbf{A} = ( A _ 1, \ldots, A _ n )$  is \emph{proportional up to one item} (PROP1) if for any $ i \in N$, there exists $e \in E \setminus A _ i $ such that $A_i\cup \{e\} \in \mathcal{S}$ and $ v _ i (A _ i \cup \{ e\}) \geq \frac{1}{n} \cdot v _ i (E)$;
	For allocating chores, the allocation is PROP1 if for any $i \in N$, there exists $ e \in A _ i $ such that $A _ i\setminus \{e\} \in \mathcal{S}$ and $ v _ i ( A _ i \setminus \{ e \}) \geq \frac{1}{n} \cdot v _ i ( E )$.
\end{definition}

An alternative relaxation of PROP is maximin share (MMS) fairness.
Given an instance $\cI$, the \emph{maximin share} (MMS) of agent $ i \in N$ is the maximum value she can guarantee if she partitions all items $E$ into $n$ connected bundles but receives the smallest one.
Formally,
$$
\MMS_i(E,n) =  \max\limits_{\mathbf{X} \in \Pi_n(E)}\min\limits_{ j \in N} v _ i ( X _ j ).
$$
If the instance $\cI$ is clear from the context, we write $\MMS_i(\cI)$ or $\MMS_i$ for simplicity. 
Moreover, it is not hard to verify that
for any fair-goods instance $\cI$, $\MMS_i(\mathcal{I}) \leq \frac{1}{n}$, and 
for any fair-chores $\cI$, $\MMS_i(\mathcal{I}) \leq \min\{-\frac{1}{n}, \min_{e\in E}v_i(e)\}$.
For any agent $i$ and a connected $n$-partition $\cA=(A_1,\ldots,A_n)$, if $v_i(A_j)\ge \MMS_i$ for all $j$, $\cA$ is called an $\MMS_i$-defining partition.
Note that although the computation of MMS values without connectivity constraints is NP-hard \cite{DBLP:journals/orl/Woeginger97}, in our setting, it can be computed efficiently.
Our proof of any result in this paper either immediately follows the statement of the result, or can be found in the
Appendix clearly indicated.

\begin{lemma}
\label{lem:mms:computation}
    The value of $\MMS_i$ can be computed in polynomial time for any agent $i\in N$.
\end{lemma}

When items are goods, Lemma~\ref{lem:mms:computation} can be implied from the results on tree established by \cite{DBLP:conf/ijcai/BouveretCEIP17}, and we give a simple proof for the case of chores.
\begin{definition} [MMS]
	An allocation $\cA=(A_1,\ldots,A_n)$ of an instance $\cI$ is \emph{maximin share }(MMS) fair if for any $i\in N$, $ v _ i ( A _ i ) \geq \mathsf{MMS}_i(\cI)$.
\end{definition}



Our algorithms in this paper imply the existence and efficient computation of a(n) PROP1 and MMS allocation for allocating contiguous blocks of both goods and chores. Since we have multiple item types, welfare objectives and fairness notions, we use  $X \mid W \mid F$ to denote each specific setting, where $X$ is either goods or chores, $W$ is either egalitarian or utilitarian welfare, and $F$ is either MMS or PROP1.

\subsection{Price of Fairness}
The \emph{price of fairness} (PoF) with respect to particular fairness and efficiency notions evaluates the economic welfare loss by enforcing the allocations to be fair.
For the allocation of goods, PoF is the supremum ratio over all instances between the maximum welfare of all allocations and maximum welfare of all fair allocations;
For the allocation of chores, PoF is the supremum ratio over all instances between the maximum welfare of all fair allocations and maximum welfare of all allocations.
Given an instance $\cI$ and a
welfare function $W \in \{ \EW, \UW \}$, denote by $\OPT_W(\cI)$
the maximum welfare with respect to $W$ over all allocations of $\cI$. 
For simplicity, if the instance is clear from the context, we write $\OPT_E$ and $ \OPT_U$ to refer $\OPT_{EW}(\mathcal{I})$ and $\OPT_{UW}(\mathcal{I})$, respectively.
Letting $F \in \{ \text{PROP1, MMS} \}$ be a fairness criterion, denote by $F(\cI)$ the set of all allocations satisfying $F$.

\begin{definition}[PoF]
	The price of fairness for fair-goods instances with respect to fairness criterion $F$ and welfare function $W$ is
	$$
	\PoF(G\mid W \mid F) = \sup\limits_{\cI} \min\limits_{\mathbf{A} \in F(\cI)} \frac{\OPT_W(\cI)}{W(\mathbf{A})}.
	$$
	The price of fairness for fair-chores instances with respect to fairness criterion $F$ and welfare function $W$ is
	$$
	\PoF (C\mid W \mid F) = \sup\limits_{\cI} \min\limits_{\mathbf{A} \in F(\cI)} \frac{W(\mathbf{A})}{\OPT_W(\cI)}.
	$$
	When the setting ($X\mid W \mid F$) is clear from the context, we simply write $\PoF$.
\end{definition}

We remark that when assigning goods, if no fair allocations can achieve a non-zero welfare, the price of fairness is infinite.
The price of fairness with respect to fairness criterion $F$ is also called {\em price of $F$}, i.e., price of PROP1 and price of MMS.


\section{Fair Division of Goods}
\label{sec:goods}
In this section, we discuss the PoF for indivisible goods. 
We start with MMS fairness in Section \ref{sec:goods:MMS} and then PROP1 fairness in Section \ref{sec:goods:PROP1}.
Particularly, we focus on the general case when $n \ge 3$ in this section and defer the case when $n=2$ to Section \ref{sec:n=2}.
For each fairness notion, we design efficient algorithms to prove the upper bound results as shown in Table \ref{table:main-result:general} and provide well-constructed instances to prove the tightness of these bounds.

\subsection{Price of MMS for Indivisible Goods}
\label{sec:goods:MMS}

\subsubsection{Useful Subroutines in the Algorithms}

We first provide two useful subroutines, $\Matching$ and $\MovingKnife$ as shown in Algorithms \ref{ALG:matching} and \ref{ALG::moving-knife}.
They and their variants will be actively used in the following algorithm design. 
Intuitively, given a fair-goods instance $\cI$, $\Matching(\cI,\calpha)$ uses a threshold vector $\calpha = (\alpha_1,\ldots,\alpha_n)$ to identify a set of large items for each agent, and then uses a maximum matching to assign large items to the agents so that each agent gets at most one item. 
Particularly, if $\alpha_i \ge \MMS_i$, then the agents who get allocated by $\Matching$ are happy with MMS fairness. 
$\MovingKnife(\cI,\cA',\calpha)$ is motivated by the moving-knife algorithm \cite{dubinsHowCutCake1961}, where we stand at the left-end of the items and find the closest item such that there exists an agent who is happy (regarding the parameters $\alpha_i$'s) with the contiguous block between left-end and this item. 
As we will see, $\MovingKnife$ can ensure MMS fairness by setting $\alpha_i = \MMS_i(\cI)$ but may produce low utilitarian and egalitarian welfare. 
We can adjust $\alpha_i$'s to increase the welfare guarantee, however, we need to be very careful because otherwise we may not find a connected allocation to satisfy all agents. 
In the following sections, we show how to choose the parameters and combine $\Matching$ and $\MovingKnife$ to satisfy the connectivity requirement and also ensure the  welfare guarantees as shown in Table \ref{table:main-result:general}.

\begin{algorithm}[H]
	\caption{\hspace{-2pt}{$\Matching(\cI,\calpha)$}}
	\label{ALG:matching}
	\begin{algorithmic}[1]
		\REQUIRE Instance $\mathcal{I} = \langle N, E, \mathcal{V} \rangle$, a parameter vector $\calpha = (\alpha_i)_{i\in N}$ where $\alpha_i \ge 0$.
		\ENSURE A partial allocation $\mathbf{A}'$ where each agent is allocated at most one item.
		\STATE Construct a weighted bipartite graph $G = (N\cup E, N \times E )$ where agents are vertices on one side and items are vertices on the other side. 
		For each $i\in N$ and $j\in E$, there is an edge $(i,j)$ with weight $v_i(e_j)$ if $v_i(e_j) \ge \alpha_i$.
		
		\STATE Compute a maximum {\em weighted} matching $\mu$ of $G$ and denote by $\mu(i)$ the item matched to agent $ i \in N $. 
		\label{step:goods:matching:weight}
		If agent $ i $ is unmatched, $\mu ( i ) =\emptyset$. 
		\STATE For each agent $i \in N$, set $A' _ i  \leftarrow \{\mu (i)\}$ if $\mu (i) \neq \emptyset$ and $A'_i  \leftarrow \emptyset$ otherwise.

		\RETURN $\mathbf{A}' = (A'_1,\ldots,A'_n)$.
	\end{algorithmic}
\end{algorithm}

	\begin{algorithm}
	\caption{\hspace{-2pt}{$\MovingKnife(\cI,\cA',\calpha)$}}
	\label{ALG::moving-knife}
	\begin{algorithmic}[1]
		\REQUIRE Instance $\mathcal{I} = \langle N, E, \mathcal{V} \rangle$, a partial allocation $\cA'$ of $Q \subseteq E$ on $P \subseteq N$ with $|P| = |Q| < n$, 
		a parameter vector $\calpha = (\alpha_i)_{i\in N\setminus P}$ where $\alpha_i \ge 0$.
		\ENSURE A complete allocation $\mathbf{A}$.
		\STATE Initialize $ \bar{N} \leftarrow N\setminus P$, $ \bar{E} \leftarrow E\setminus Q $, $ p _ 0 \leftarrow 0 $;
		\WHILE{$|\bar{N} | \geq 1 \And \bar{E}  \neq \emptyset$}
		\STATE Find the smallest index $p > p _ 0 $ such that there is an agent $ i $ with value $ v _ i ( S ) \geq \alpha_i$ on a connected bundle $ S \subseteq L (p) \cap \bar{E} $. 
		If there are multiple agents, choose the one with the largest value on $S$.
		If no such $p$ is found, break and return ``{\bf Error}''.
		\label{step:movingknife:error}
		\STATE Set $ A _ i \leftarrow S $ and update $ \bar{N} \leftarrow \bar{N}  \setminus \{ i \} $, $ \bar{E}  \leftarrow \bar{E} \setminus A _ i $, $ p _ 0 \leftarrow p $.
		\ENDWHILE
		\RETURN $\mathbf{A} = (A_1,\ldots,A_n)$.
	\end{algorithmic}
\end{algorithm}

\subsubsection{$\text{Goods}\mid\text{Utilitarian}\mid\text{MMS}$}

We first design an algorithm (as shown in Algorithm \ref{ALG::MMS-good-utilitarian}) to compute an MMS fair allocation that ensures constant utilitarian welfare, 
which implies the $O(n)$ upper bound since the valuations are normalized and the utilitarian welfare of any allocation is no more than $n$. 
We then complement this result by showing that no algorithm can ensure better than $\Theta(\frac{1}{\sqrt{n}})$ fraction of the optimal utilitarian welfare.
To design such an algorithm, our first failed trail is to run $\MovingKnife$ by setting $\alpha_i=\max\{\MMS_i,\frac{1}{4n}\}$ for all agents $i$ and $\cA' = \emptyset$, so that every agent is satisfied regarding MMS fairness and has value no smaller than $\frac{1}{4n}$ which implies constant utilitarian welfare.
However, by setting $\alpha_i=\max\{\MMS_i,\frac{1}{4n}\}$, $\MovingKnife$ may not return a feasible allocation.
For example, consider an instance with $N=\{1,2\}$, $E=\{e_1,e_2\}$, and both agents have value $\epsilon \ll \frac{1}{8}$ for $e_1$ and $1-\epsilon$ for $e_2$.
Since $\MMS_i = \epsilon$ and $\alpha_i = \frac{1}{8}$, $\MovingKnife$ allocates both $e_1$ and $e_2$ to one of the agents and the other agent gets nothing.
The failure of this trail is because there is one item, i.e., $e_2$ in the example, that is too large so that by allocating it together with some items on its left to one agent, 
there may not be enough items left to ensure MMS values for the remaining agents. 

	\begin{algorithm} 
	\caption{\hspace{-2pt} Goods $\mid$ Utilitarian $\mid$ MMS}
	\label{ALG::MMS-good-utilitarian}
	\begin{algorithmic}[1]
		\REQUIRE instance $\mathcal{I} = \langle N, E, \mathcal{V} \rangle$.
		\ENSURE A connected allocation $\mathbf{A}$.
		\STATE Run Subroutine $\Matching(\cI,\calpha)$ with $\alpha_i = \max\{\MMS_i(\cI), \frac{1}{4n}\}$ and obtain a partial allocation $\cA'$ on agents $N_0\subseteq N$ and items $E_0 \subseteq E$.
		Let $N ^{\prime} = N \setminus N _ 0 $, $ E ^ {\prime} = E \setminus E _ 0 $.
		
		\IF {$\UW(\mathbf{A}') \geq \frac{1}{4}$}
		\STATE Run $\MovingKnife(\cI,\cA',\calpha)$ with $\alpha_i = \MMS_i(\cI)$ for $i \in N'$ and obtain allocation $\cA$. \label{step:goods:u:mms:1}
		\ELSE
		\STATE  Run $\MovingKnife(\cI,\cA',\calpha)$ with $\alpha_i = \max\{\MMS_i(\cI),\frac{1}{4n}\}$ for $i \in N'$ and obtain allocation $\cA$. \label{step:goods:u:mms:2}
		\ENDIF
		\STATE If there are unallocated connected items, arbitrarily assign them to the agent whose bundle is connected with them. %
		\RETURN $\mathbf{A}$
	\end{algorithmic}
\end{algorithm}

Therefore, in Algorithm \ref{ALG::MMS-good-utilitarian}, we first use $\Matching$ to allocate large items by setting each $\alpha_i = \max\{\MMS_i(\cI),\frac{1}{4n}\}$.
Note that if the partial utilitarian welfare of the agents who get allocated by $\Matching$ is at least $\frac{1}{4}$, then it suffices to ensure solely MMS values for the remaining agents and $\MovingKnife$ with $\alpha_i=\MMS_i$ is able to find such a connected allocation. 
If the welfare from $\Matching$ is smaller than $\frac{1}{4}$, to ensure constant total utilitarian welfare, we turn to feed $\alpha_i=\max\{\MMS_i,\frac{1}{4n}\}$ to $\MovingKnife$ so that every agent's value is at least $\frac{1}{4n}$, which implies constant welfare as well. 
Fortunately, given the welfare from $\MovingKnife$ being smaller than $\frac{1}{4}$, we manage to prove that the previous failure will not happen and we are able to allocate each remaining agent a connected block with value no smaller than $\max\{\MMS_i,\frac{1}{4n}\}$. 
We have the main result as follows. 

\begin{theorem}\label{thm::MMS-good-UW-n>3}
	For Goods $\mid$ Utilitarian $\mid$ MMS, $\PoF$ is at least $\Omega(\sqrt{n})$ and at most $ O (n)$. 
\end{theorem}

The upper bound in Theorem \ref{thm::MMS-good-UW-n>3} relies on the following two technical lemmas which show that Algorithm \ref{ALG::MMS-good-utilitarian} can always return a feasible allocation with utilitarian welfare at least $\frac{1}{4}$.



\begin{lemma}\label{lemma::MMS-good-U-add-1}
	For any instance $\cI=\langle N, E, \mathcal{V} \rangle$, if the allocation $\mathbf{A}$ returned by Algorithm~\ref{ALG::MMS-good-utilitarian} is constructed in Step \ref{step:goods:u:mms:1}, 
	then $v_i(A_i) \ge \MMS_i$ for agent $i\in N$.
\end{lemma}

\noindent\begin{proof}
    The correctness of Lemma \ref{lemma::MMS-good-U-add-1} relies on the following claim, which shows that if we remove a connected set of items (starting from the left end) whose value is small enough together with one agent, every agent's MMS value does not decrease in the reduced instance.
    
    \begin{claim}\label{claim:reduction}
    For any agent $i$, let $p\le m$ be the smallest index such that $v _ i (L(p)) \ge \MMS_i(\cI)$, then $\MMS_i(E \setminus L(p), n-1) \ge \MMS_i(\cI)$.
    \end{claim}
    
    \noindent\begin{proof}[Proof of Claim \ref{claim:reduction}]
        Consider an arbitrary $\MMS_i(\cI)$-partition $(X_1,\ldots,X_n)$ for agent $i$, where the bundles are ordered from left to right.
        Then $v_i(X_j) \ge \MMS_i(\cI)$ for any bundle $X_j$.
        The condition of $p$ being the smallest index such that $v_i(L(p)) \ge \MMS_i(\cI)$ implies that $L(p) \subseteq X_1$ and $E \setminus L(p) \supseteq \bigcup_{j=2}^n X_j$.
        Note that $(X_2,\dots,X_n)$ is one $(n-1)$-partition of $\bigcup_{j=2}^n X_j$,
        and thus 
        \[
        \MMS_i(E \setminus L(p), n-1) \ge \MMS_i(\bigcup_{j=2}^n X_j, n-1)\ge \MMS_i(\cI),
        \]
        which completes the proof of the claim.
    \end{proof}
    
    Note that the $\Matching$ procedure can be ideologically mingled with the $\MovingKnife$ procedure where an agent takes away a bundle with a single item.
    Without loss of generality, assume the bundles are allocated from left to right, and it is clear that if an agent takes away a bundle, her value is no smaller than her MMS.
    Let $E_0$ be the remaining items at the beginning of some round and suppose agent $i$ takes away $L(p) \cap E_0$ in the current round.
    For any other remaining agent $j$, let $p_j$ be the smallest index such that $v_j(L(p_j) \cap E_0) \ge \MMS_j$, then by the design of the algorithm, $p \le p_j$, which is true no matter agent $i$ gets $L(p) \cap E_0$ from $\Matching$ or $\MovingKnife$.
    Then by Claim \ref{claim:reduction}, in the reduced instance she can still be satisfied regarding MMS.
    By induction, we have the lemma.
    \end{proof}

\begin{lemma}\label{lemma::MMS-goods-U-add-2}
	For any instance $\cI=\langle N, E, \mathcal{V} \rangle$, if the allocation $\mathbf{A}$ returned by Algorithm~\ref{ALG::MMS-good-utilitarian} is constructed in Step \ref{step:goods:u:mms:2}, 
	then $v_i(A_i) \ge \max\{\MMS_i,\frac{1}{4n}\}$ for agent $i\in N$.
\end{lemma}



According to Lemmas~\ref{lemma::MMS-good-U-add-1} and \ref{lemma::MMS-goods-U-add-2}, if the allocation $\mathbf{A}$ returned 
Algorithm~\ref{ALG::MMS-good-utilitarian} is constructed in Step \ref{step:goods:u:mms:1}, it is straightforward that $\UW(\cA) \ge \frac{1}{4}$; 
if $\mathbf{A}$ is constructed in Step \ref{step:goods:u:mms:2}, 
every agent receives a connected bundle with value at least $\frac{1}{4n}$ and thus we also have $\UW(\mathbf{A}) \geq \frac{1}{4}$.
Since the utilitarian welfare of any connected allocation is at most $n$, the price of MMS with respect to utilitarian welfare is at most $ 4n$.
The lower bound in Theorem \ref{thm::MMS-good-UW-n>3} is because of the following lemma.

\begin{lemma}\label{lem:goods:u:mms:lb}
No algorithm can ensure better than $\Theta(\frac{1}{\sqrt{n}})$ fraction of the optimal utilitarian welfare using MMS fair allocations.
\end{lemma}

In addition, our proof of Lemma~\ref{lemma::MMS-good-U-add-1} indeed results in a stronger property regarding MMS allocations, which may be independent of interest. Informally, if we assign $k$ arbitrary items to $k$ arbitrary agents we can still allocate the remaining items to remaining agents which is feasible and guarantees MMS for them.

\begin{corollary}\label{lem:pre:reduction}
For any instance $\cI$, any $P \subseteq N$ and $Q \subseteq E$ with $|P| = |Q| = k \le n$, there exists $(X_i)_{i\in N\setminus P} \in \Pi_{n-k}(E\setminus Q)$ such that $v_i(X_i) \ge \MMS_i(\cI)$ and $X_i \in \cS$ for all $i \in N\setminus P$.
\end{corollary}

\subsubsection{$\text{Goods}\mid\text{Egalitarian}\mid\text{MMS}$}

Next, we discuss the egalitarian welfare. 
For any instance $\cI$, if $\OPT_E(\mathcal{I}) \ge \frac{1}{n}$, 
it means that the egalitarian welfare maximizing allocation guarantees MMS fairness for every agent since $\MMS_i\leq\frac{1}{n}$ for all $i$.
Thus it suffices to consider the case where $\OPT_E(\mathcal{I}) < \frac{1}{n}$, and we show that we can find an MMS allocation by Algorithm \ref{ALG::MMS-good-egalitarian} where every agent's value is no smaller than $\Theta(\frac{1}{n})$ fraction of $\OPT_E(\mathcal{I})$.
Intuitively, agents with small MMS values are in disadvantage when we use MMS allocations to approximate $\OPT_E$ 
since when such an agent receives a bundle with value MMS we may stop assigning more items to her.
Therefore, Algorithm \ref{ALG::MMS-good-egalitarian} first allocates large items to agents with small MMS values via matchings like Subroutine $\Matching$.
Particularly, 
we set $\alpha_i = + \infty$ if $\MMS_i(\cI) \ge \frac{1}{2n}\cdot\OPT_E$ and $\alpha_i = \frac{1}{2n} \cdot \OPT_E$ otherwise. 
Another difference with $\Matching$ is that the maximum weighted matching used in Step \ref{step:goods:matching:weight} of $\Matching$ is changed to maximum {\em cardinality} matching and all the other steps are the same.
We denote the revised algorithm by $\Matching'$, and for completeness, a formal description is provided in the appendix; see Algorithm \ref{ALG:matching:cardinality}. 
Thereafter, Algorithm \ref{ALG::MMS-good-egalitarian} uses $\MovingKnife$ to allocate the remaining items by properly choosing the parameter. 
Seemingly we are quite conservative on the selection of the parameter $\frac{1}{2n} \cdot \OPT_E$, but it turns out to induce the (asymptotically) best possible PoF ratio.


	\begin{algorithm} 
	\caption{\hspace{-2pt} Goods $\mid$ Egalitarian $\mid$ MMS}
	\label{ALG::MMS-good-egalitarian}
	\begin{algorithmic}[1]
		\REQUIRE An instance $\mathcal{I} = \langle N, E, \mathcal{V} \rangle$ with $\textnormal{OPT}_E < \frac{1}{n}$.
		\ENSURE A connected allocation $\mathbf{A}$.
		\STATE Run Subroutine $\Matching'(\cI,\calpha)$ with $\alpha_i = + \infty$ if $\MMS_i(\cI) \ge \frac{1}{2n} \cdot \OPT_E$ and $\alpha_i = \frac{1}{2n} \cdot \OPT_E$ otherwise, and obtain a partial allocation $\cA'$ on agents $N_0\subseteq N$ and items $E_0 \subseteq E$.
		Let $N ^{\prime} = N \setminus N _ 0 $, $ E ^ {\prime} = E \setminus E _ 0 $.  \label{step:goods:e:mms:1}


		
		\STATE  Run $\MovingKnife(\cI,\cA',\calpha)$ with $\alpha_i = \max\{\MMS_i(\cI),\frac{1}{2n}\OPT_E\}$ for $i \in N'$ and obtain allocation $\cA$. \label{step:goods:e:mms:2}
		
		\STATE  If there are unallocated connected items, arbitrarily assign them to the agent whose bundle is connected with them.
		
		\RETURN $\mathbf{A}$

	\end{algorithmic}
\end{algorithm}

\begin{theorem}\label{thm::MMS-good-RW-n>3}
	For $\text{Goods}\mid\text{Egalitarian}\mid\text{MMS}$ problem, $\text{PoF} = \Theta(n)$.
\end{theorem}

To prove Theorem \ref{thm::MMS-good-RW-n>3}, it suffices to prove Lemmas \ref{lem:goods:e:mms:ub} and \ref{lem:goods:e:mms:lb}. 

\begin{lemma}\label{lem:goods:e:mms:ub}
For any instance $\cI=\langle N, E, \mathcal{V} \rangle$, Algorithm \ref{ALG::MMS-good-egalitarian} returns an allocation $\cA$ where $v_i(A_i) \ge \max\{\MMS_i(\cI), \frac{1}{2n}\cdot \OPT_E\}$ for all agents $i \in N$. 
\end{lemma}

Note that Algorithm \ref{ALG::MMS-good-egalitarian} uses the value of $\OPT_E$ which is NP-hard to compute \cite{DBLP:conf/ijcai/BouveretCEIP17}.\footnote{It is proved in \cite{DBLP:conf/ijcai/BouveretCEIP17} that deciding whether an instance admits a PROP allocation (i.e., whether it admits an allocation with egalitarian welfare no smaller than $\frac{1}{n}$) is NP-hard, which is a special case of our problem.}
To make it run in polynomial time, we adopt the following technique by guessing the value of $\OPT_E$, denoted by $o$.
If $o \le \OPT_E$, by Lemma \ref{lem:goods:e:mms:ub}, we can find a feasible allocation where $v_i(A_i) \ge \max\{\MMS_i(\cI), \frac{1}{2n}\cdot o\}$.
Thus, we can start with a trivial upper bound of  $\OPT_E$ by setting $o=1$, and run Algorithm \ref{ALG::MMS-good-egalitarian}. 
If we do not obtain a feasible allocation, we decrease the value of $o$ by a small constant $\epsilon > 0$ and repeat Algorithm \ref{ALG::MMS-good-egalitarian}.
We stop until we obtain a feasible allocation regarding $o$ and it is guaranteed that  $o \ge \OPT_E - \epsilon$.
To ensure $\OPT_E - \epsilon = \Theta(\frac{1}{n}) \cdot \OPT_E$ for all agents, the value of $\epsilon$ cannot be too large. 
One possible way is to set $\epsilon = \frac{1}{2} \cdot \min_{i \in N} \MMS_i$, since $\OPT_E \ge \min_{i \in N} \MMS_i$ and thus $ \OPT_E - \epsilon \ge \frac{1}{2} \cdot \OPT_E$.


\begin{lemma}\label{lem:goods:e:mms:lb}
No algorithm can ensure better than $\Theta(\frac{1}{n})$ fraction of the optimal egalitarian welfare using MMS fair allocations.  
\end{lemma}
\begin{proof}[Proof of Lemma \ref{lem:goods:e:mms:lb}]
    Consider the following instance $\cI = \langle N, E, \mathcal{V} \rangle$ with $n \ge 3$ agents and $kn+1$ items, where $k \gg n$ is an integer. 
    The valuations are shown in the following table. 
    	\begin{table}[h]
	    \centering
	    \begin{tabular}{c|c|c|c|c|c|c}
	         Items & $e_1$ & $e_2$ & $\cdots$ & $e_{kn-1}$ & $e_{kn}$ & $e_{kn+1}$  \\
	         \hline
	         $v_i(\cdot)$ for $i =1,\ldots, n-2$ & $\frac{1}{n\cdot k}$ & $\frac{1}{n\cdot k}$ & $\cdots$ & $\frac{1}{n\cdot k}$ & $\frac{1}{n\cdot k}$ & 0\\
	         $v_i(\cdot)$ for $i= n-1$ and $n$ & $\frac{1}{k\cdot n^{k+1}}$ & $\frac{1}{k\cdot n^{k+1}}$ & $\cdots$ & $\frac{1}{k\cdot n^{k+1}}$ & $\frac{1}{k\cdot n^{k+1}}$ & $ 1-\frac{1}{n^{k}}$   \\
	    \end{tabular}
	    \caption{The Lower Bound Instance in Lemma \ref{lem:goods:e:mms:lb}}
	\end{table}
    
    On one hand, it is not hard to verify that $\MMS_i = \frac{1}{n}$ for $ i = 1, \ldots, n-2$ and $\MMS_i < \frac{2}{n^{k+1}}$ for $i=n-1$ and $n$. 
    In an MMS allocation, each agent $ i \in [n-2]$ has to receive at least $ k $ items from $\{e_1,\ldots,e_{kn}\}$ and 
    hence one of agents $ i =n-1, n$ cannot obtain item $kn+1$ and thus has value at most $2k \cdot \frac{1}{k n^{k+1}}$.
    That is, an MMS allocation has egalitarian welfare at most $ \frac{2}{n ^ {k + 1}}$. 
    On the other hand, consider the allocation $\mathbf{O}$ with $ O _i = \{ e _ i \}$ for $i\in [n-2]$, $O_{n-1} = \{e_{n-1},\ldots,e_{kn}\}$, and $O_n = \{e_{kn+1}\}$.
    We have $v_i(O_i) = \frac{1}{n k}$ for $i\in [n-2]$, $v_{n-1}(O_{n-1}) = \frac{kn-n+2}{k n ^ {k+1}}$ and $v_n(O_n) = 1-\frac{1}{n^k}$, which means $\OPT_E \ge \EW(\mathbf{O}) = \frac{kn-n+2}{k \cdot n ^ {k+1}}$.
    Therefore, 
    \[
    \PoF \ge \frac{kn-n+2}{2k} \rightarrow \frac{n}{2} \text{ when $k \rightarrow +\infty$,}
    \]
    which completes the proof of the lemma.
\end{proof}

\subsection{Price of PROP1 for Indivisible Goods}

\label{sec:goods:PROP1}

\subsubsection{$\text{Goods}\mid\text{Utilitarian}\mid\text{PROP1}$}

In this section, we will prove the following result. 

\begin{theorem}\label{thm::PROP1-good-UW-n>3}
For $\text{Goods}\mid\text{Utilitarian}\mid\text{PROP1}$, $\PoF$ is 
at least $\Omega(\sqrt{n})$ and at most $O(n)$.
\end{theorem}

Again, we split the proof of Theorem \ref{thm::PROP1-good-UW-n>3} into proving the upper and lower bounds as shown by the following Lemmas \ref{lemma::PROP1-good-UW-add-1} and \ref{lem:goods:u:prop1:lb}. 


\begin{lemma}\label{lemma::PROP1-good-UW-add-1}
	A \textnormal{PROP1} allocation with utilitarian welfare at least $\frac{1}{2}$ can be computed efficiently.
\end{lemma}
Since the valuations are normalized, the utilitarian welfare of any allocation is no larger than $n$, and thus Lemma \ref{lemma::PROP1-good-UW-add-1} implies the upper bound of $O(n)$.
We provide some intuitions below.
To show there is a PROP1 allocation with utilitarian welfare at least $\frac{1}{2}$, 
we first observe that if an agent has value no less than $\frac{1}{n}$ for some item, then by definition, we do not need to allocate any item to her and PROP1 is trivially satisfied. 
Let $N_1$ be the set of such agents.
If $|N_1| = n$, by assigning all items to an arbitrary agent, we have utilitarian welfare at least $1$.
If $|N_1| = n-1$, by assigning all items to the remaining agent, we have utilitarian welfare at least $1$.
Thus it suffices to consider the case when $N_2 = N \setminus N_1$ and $|N_2| \ge 2$.
Consider a relaxed problem restricted on $N_2$ when the items are divisible, i.e., cake cutting problem. 
It is shown in \cite{dubinsHowCutCake1961} that there is a contiguous fractional allocation $\cA$ that is PROP which means the utilitarian welfare is no smaller than 1.  
Note that in allocation $\cA$, no item can be allocated to more than two agents since all agents in $N_2$ have value smaller than $\frac{1}{n}$ for any single item.
To convert $\cA$ into an integral allocation, we consider two ways by allocating each fractional item to its left or right agent. 
Denote by $\cA^L$ and $\cA^R$ the corresponding integral allocations. 
It can be proved that both of them are contiguous PROP1, and $\UW(\cA^L) + \UW(\cA^R) \ge 1$, 
which means that one of $\cA^L$ and $\cA^R$ has utilitarian welfare at least~$\frac{1}{2}$.

\begin{lemma}\label{lem:goods:u:prop1:lb}
No algorithm can ensure better than $\Theta(\frac{1}{\sqrt{n}})$ fraction of the optimal utilitarian welfare by PROP1 allocations. 
\end{lemma}

\subsubsection{$\text{Goods}\mid\text{Egalitarian}\mid\text{PROP1}$}
\begin{theorem}\label{thm::PROP1-good-RW-n>3}
For $\text{Goods}\mid\text{Egalitarian}\mid\text{PROP1}$, $\PoF= \infty$.
\end{theorem}

\noindent\begin{proof}
Consider an instance with $n \ge 3$ agents and $n+1$ items $ E = \{ e _ 1, \ldots, e _ { n + 1 } \}$. 
Agents have identical valuation functions, where $ v _ i ( e _ j ) = \frac{1}{n^2}$ for $ j =1,\ldots, n$ and $ v _ i ( e _ { n + 1 }) = 1 - \frac{1}{n}$. 
It is straightforward that the optimal egalitarian welfare is at least $\frac{1}{n^2}$. 
There are three ways to satisfy PROP1 for each agent $i$:
(1) obtaining $e_{n+1}$ (possibly with some items on its left) so that PROP is satisfied, 
(2) obtaining $e_n$ (possibly with some items on its left ) so that $v_i(\{e_n, e_{n+1}\})\ge \frac{1}{n}$,
(3) obtaining nothing so that $v_i(\{e_{n+1}\}) \ge \frac{1}{n}$.
Thus, in any PROP1 allocation, at least $n-2$ agents receive nothing and achieve PROP1 via (3), which means the egalitarian welfare is 0 and the price of PROP1 is infinite.
\end{proof}

\paragraph{Remark}
Note that the hard instance designed in Theorem \ref{thm::PROP1-good-RW-n>3} shows a significant difference between MMS and PROP1.
If an agent cannot receive items $e_n$ or $e_{n+1}$, with MMS, she wants to have one item from $\{e_1,\ldots,e_{n-1}\}$;
however, with PROP1, she actually prefers to receive nothing due to the connectivity requirement in the definition of PROP1.

\section{Fair Division of Indivisible Chores}

\label{sec:chores}
In this section, we briefly discuss the results on indivisible chores.
Again, we focus on the general case of $n \ge 3$ and the case with two agents is discussed in Section \ref{sec:n=2}.

\subsection{Price of MMS for Indivisible chores}

We start with MMS fairness, and show that a single polynomial time algorithm (see Algorithm \ref{alg:chores:MMS}) gives the tight ratio of the price of MMS regarding both utilitarian and egalitarian welfare.   



\begin{algorithm} 
	\caption{\hspace{-2pt}{ Chores $\mid$ MMS}}
	\label{alg:chores:MMS}
	\begin{algorithmic}[1]
		\REQUIRE A fair-chores instance $I = \langle N, E, \mathcal{V} \rangle$.
		\ENSURE Allocation $\mathbf{A} = (A_1,\ldots,A_n)$.
		\STATE Initialize $N_0 \leftarrow N$ and $ E _ 0 \leftarrow E$.
		\WHILE{$|N _ 0 | > 1 \And E_ 0 \neq \emptyset$}
		\STATE Denote by $e _ L \in E _ 0 $ the left-most item in $E _ 0$.
		\IF {there exists $ i \in N _ 0$ such that $ v _ i (e _ L ) \geq \max\{\MMS_i, -\frac{2}{n}\}$}
		\label{step:chores:mms:1:con}
		\STATE Let $p$ be the largest index such that there exists an agent $i$ with $ v _ i ( L (p) \cap 
		E _ 0 ) \geq\max\{ \textnormal{MMS} _ i , -\frac{2}{n}\}$. 
		If there is a tie, select the agent $i$ with largest value on $L (p) \cap E _ 0$.
		\STATE $A _ i \leftarrow L(p)\cap E _ 0 $, $E _ 0 \leftarrow E _ 0 \setminus A _ i $ and $ N _ 0 \leftarrow N _ 0 \setminus \{ i \} $.\label{step:chores:mms:1}
		\ELSE
		\STATE Let $ i \in \arg\max_{j \in N_ 0} v _ j (e _ L)$, breaking tie arbitrarily.
		\label{step:chores:mms:2:con}
		\STATE $A _ i \leftarrow \{ e _ L \}$, $ E _0 \leftarrow E _ 0 \setminus \{ e _ L \} $ and $ N _ 0 \leftarrow N_ 0 \setminus \{ i \}$. \label{step:chores:mms:2}
		\ENDIF
		\ENDWHILE

		\IF{$E _ 0 \neq \emptyset$}
		\STATE Let $l$ be the remaining agent in $N_0$.
		\STATE $ A _ {l} \leftarrow E _ 0 $. \label{step:chores:mms:3}
		\ENDIF
		
		\RETURN $\mathbf{A}$
	\end{algorithmic}
\end{algorithm}

Algorithm \ref{alg:chores:MMS} also utilizes the idea of moving knife like Algorithm \ref{ALG::moving-knife}, where we stand at the left end of the items and repeatedly find the {\em farthest} item such that there is an agent whose value for the items from the left end to the current one is no larger than $\max\{\MMS_i, -\frac{2}{n}\}$.
The intuition of the parameter $\max\{\MMS_i, -\frac{2}{n}\}$ is to ensure that the agents with large MMS values do not receive too many items which may induce low welfare.
In the following lemma, we show that Algorithm \ref{alg:chores:MMS} can allocate all items in the way that all agents are happy regarding MMS.

\begin{lemma}\label{Lemma::MMS-UW-chore}
	Algorithm \ref{alg:chores:MMS} returns an MMS allocation in polynomial time.
\end{lemma}

The advantage of Algorithm \ref{alg:chores:MMS} is that it returns allocations that either
directly achieves the best possible ratio of price of MMS for both utilitarian and egalitarian welfare at the same time or leads us to MMS allocations with the desired welfare guarantee.

\begin{theorem}\label{thm::MMS-chore-UW}
	For the problem of Chores $\mid$ Utilitarian $\mid$ MMS, $\PoF = \Theta(n)$.
\end{theorem}



\begin{theorem}\label{thm::MMS-chore-RW}
    For the problem of Chores $\mid$ Egalitarian $\mid$ MMS, $\PoF = \frac{n}{2}$.
\end{theorem}


\subsection{Price of PROP1 for Indivisible chores}

For PROP1 fairness, we also have a single algorithm computing allocations that either guaranteeing the tight ratio of price of PROP1 regarding both utilitarian and egalitarian welfare or lead us to the PROP1 allocations with the desired welfare guarantee.
We provide the detailed algorithm and the proofs of the theorems in the appendix.

\begin{theorem}\label{thm::PROP1-chore-UW-n>3}
For $\text{Chores}\mid\text{Utilitarian}\mid\text{PROP1}$, $\PoF= \Theta(n)$.
\end{theorem}

\begin{theorem}\label{thm::PROP1-chore-RW-n>3}
For $\text{Chores}\mid\text{Egalitarian}\mid\text{PROP1}$, $\PoF = 2$ if $n=3$ and $\PoF=\frac{n}{2}$ if $n\geq 4$.
\end{theorem}

\section{Two agents}
\label{sec:n=2}

In this section, we focus on the allocation problem with two agents which already yields non-trivial and valuable results. 
We show that there exists a single allocation that is simultaneously MMS and PROP1 and achieves the corresponding optimal ratio of PoF.
Moreover, if the instance admits an EF or PROP allocation, this allocation is also EF or PROP and achieves the ratio of PoF proved  in \cite{suksompongFairlyAllocatingContiguous2019} and \cite{hohneAllocatingContiguousBlocks2021}.


\subsection{Goods}
\begin{lemma}\label{lem:n=2:goods}
    For any fair-goods instance with two agents, there exists an allocation that is simultaneously MMS and PROP1. Moreover, it maximizes the egalitarian welfare and achieves utilitarian welfare at least 1.
\end{lemma}
	
\noindent\begin{proof}
    We explicitly construct such an allocation $\mathbf{O} = (O_1,\ldots, O_n)$:
    \begin{itemize}
        \item $\mathbf{O}$ first maximizes the egalitarian welfare among all connected allocations; If there is a tie, $\mathbf{O}$ maximizes the number of items allocated to the agent with smaller value.
    \end{itemize}
    By the construction, it is straightforward that $\mathbf{O}$ maximizes the egalitarian welfare. 
    Without loss of generality, we assume $ v _ 1 ( O _ 1) \leq v _ 2 ( O _ 2 )$ and $ O _1$ is on the left of $ O _ 2 $. 
    If $ v _ 1 ( O _ 1 ) \geq \frac{1}{2}$, allocation $\mathbf{O}$ satisfies the conditions described in the statement. 
    If $ v _ 2 ( O _ 2 ) < \frac{1}{2}$ or $v_2(O_1) > v_1(O_1)$, swapping the bundles would increase the  egalitarian welfare. 
    In the following, it suffices to consider the case where $v_2(O_1) \le v _ 1 ( O _ 1 ) < \frac{1}{2} \leq v _ 2 ( O _ 2 )$
    and show $O_1$ satisfies MMS and PROP1 for agent 1.

	
	
	First, we show $\mathbf{O}$ is PROP1 for agent 1.
	Let $e^* \in O _ 2$ be the item such that $O _ 1 \cup\{ e ^* \}$ is connected. 
	Consider another connected allocation $\mathbf{O}^{\prime} = (O ^ {\prime} _ 1, O ^ {\prime} _ 2 )$ with $ O^{\prime} _ 1 = O _ 1\cup \{ e ^* \}$ and $ O^{\prime} _ 2 = O _ 2 \setminus \{e^* \}$. 
	Note that $v_2(O'_2) < \frac{1}{2}$ and $v_2(O'_1) > \frac{1}{2}$; otherwise, since $\mathbf{O}$ maximizes $|O_1|$, $e^*$ should be allocated to agent 1. 
	If $v_1(O'_1) \le \frac{1}{2}$, then $v_1(O'_2) \ge \frac{1}{2}$, which means by allocating $O'_2$ to agent 1 and $O'_1$ to agent 2 induces egalitarian welfare at least $\frac{1}{2}$, which is a contradiction. 

    Next, we show $\mathbf{O}$ is MMS fair for agent 1. 
	Let $\mathbf{T} = \{ T _ 1, T _ 2 \}$ be an $\MMS_1$-defining partition and $ T _ 1$ is on the left of $ T _ 2$. 
	Again we consider allocation $\mathbf{O}^{\prime} = (O'_1, O'_2)$ constructed in the previous paragraph. 
	If $ v _ 1 ( O _ 1 ) < \MMS_1$, then we have $ O _ 1 \subsetneq T _ 1 $ and thus $ O ^ {\prime} _ 1 \subseteq T _ 1  $. 
	Since $v_1(O'_1) > \frac{1}{2}$, we have $ v _ 1 ( T _ 1 ) \geq v _ 1 ( O ^ {\prime} _ 1 ) > \frac{1}{2}$ and $\MMS_1 = v_1(T_2) \leq v _ 1 ( O ^ {\prime} _ 2 )$. 
	Thus $v _ 1 ( O ^ {\prime} _ 2 ) > v_1(O_1)$.
	That is by allocating $O'_2$ to agent 1 and $O'_1$ to agent 2, we have $v_2(O'_1) > \frac{1}{2}> v_1(O_1)$ and $v _ 1 ( O ^ {\prime} _ 2 ) > v_1(O_1)$, 
	which induces higher egalitarian welfare than $\mathbf{O}$ and leads to a contradiction. 


	

	Finally, we show $\UW(\mathbf{O}) \ge 1$.
	Since $v_2(O_1) \le v_1(O_1)$, then $v_2(O_2) = 1 - v_2(O_1) \ge 1 - v_1(O_1)$ and thus $v_1(O_1) + v_2(O_2) \ge 1$, which completes the proof of the lemma.
\end{proof}

By Lemma \ref{lem:n=2:goods}, we have the following theorem whose proof is deferred to the appendix. 
\begin{theorem}
\label{thm:n=2:goods}
For allocating indivisible goods, $\PoF=1$ regarding the egalitarian welfare for both MMS and PROP1;
$\PoF=\frac{3}{2}$ regarding the utilitarian welfare for both MMS and PROP1.
\end{theorem}

\subsection{Chores}

Similarly, we have the following results for chores. 

\begin{lemma}\label{lem:n=2:chores}
	For any fair-chores instance with two agents, there exists an allocation that is simultaneously MMS and PROP1. Moreover, it maximizes the egalitarian welfare and achieves utilitarian welfare at least -1.
\end{lemma}

\begin{theorem}
\label{thm:n=2:chores}
For allocating indivisible chores, $\PoF=1$ regarding the egalitarian welfare for both MMS and PROP1;
$\PoF=2$ regarding the utilitarian welfare for both MMS and PROP1.
\end{theorem}

\section{Conclusion}
In the work, we investigated the tight ratios of price of fairness regarding both indivisible goods and chores, utilitarian and egalitarian welfare, and MMS and PROP1 fairness. 
Our work is motivated by the research of price of fairness regarding EF and PROP, but such fair allocations rarely exist.
We gave the tight ratios for all settings except the case of goods with utilitarian welfare, where the upper and lower bounds for both MMS and PROP1 are $O(n)$ and $\Omega(\sqrt{n})$ respectively.
An immediate open problem is to explore the tight ratio for this setting.
Some interesting future directions include considering alternative fairness notions like equability up to one item \cite{sunEquitabilityWelfareMaximization2022}, extending the line structure to other graphs, and studying the case when the items are mixture of goods and chores \cite{DBLP:journals/aamas/AzizCIW22}.

\bibliography{mylibrary}

\begin{thebibliography}{10}

\bibitem{DBLP:journals/tcs/AmanatidisBFHV21}
Georgios Amanatidis, Georgios Birmpas, Aris Filos{-}Ratsikas, Alexandros
  Hollender, and Alexandros~A. Voudouris.
\newblock Maximum nash welfare and other stories about {EFX}.
\newblock {\em Theor. Comput. Sci.}, 863:69--85, 2021.

\bibitem{DBLP:journals/corr/abs-2202-07551}
Georgios Amanatidis, Georgios Birmpas, Aris Filos{-}Ratsikas, and Alexandros~A.
  Voudouris.
\newblock Fair division of indivisible goods: {A} survey.
\newblock {\em CoRR}, abs/2202.07551, 2022.

\bibitem{DBLP:journals/aamas/AzizCIW22}
Haris Aziz, Ioannis Caragiannis, Ayumi Igarashi, and Toby Walsh.
\newblock Fair allocation of indivisible goods and chores.
\newblock {\em Auton. Agents Multi Agent Syst.}, 36(1):3, 2022.

\bibitem{DBLP:journals/corr/abs-2202-08713}
Haris Aziz, Bo~Li, Herv{\'{e}} Moulin, and Xiaowei Wu.
\newblock Algorithmic fair allocation of indivisible items: {A} survey and new
  questions.
\newblock {\em CoRR}, abs/2202.08713, 2022.

\bibitem{DBLP:conf/focs/AzizM16}
Haris Aziz and Simon Mackenzie.
\newblock A discrete and bounded envy-free cake cutting protocol for any number
  of agents.
\newblock In {\em {FOCS}}, pages 416--427. {IEEE} Computer Society, 2016.

\bibitem{conf/aaai/AzizRSW17}
Haris Aziz, Gerhard Rauchecker, Guido Schryen, and Toby Walsh.
\newblock Algorithms for max-min share fair allocation of indivisible chores.
\newblock In {\em {AAAI}}, pages 335--341, 2017.

\bibitem{DBLP:conf/wine/BarmanB020}
Siddharth Barman, Umang Bhaskar, and Nisarg Shah.
\newblock Optimal bounds on the price of fairness for indivisible goods.
\newblock In {\em {WINE}}, volume 12495 of {\em Lecture Notes in Computer
  Science}, pages 356--369. Springer, 2020.

\bibitem{DBLP:conf/sigecom/BarmanKV18}
Siddharth Barman, Sanath~Kumar Krishnamurthy, and Rohit Vaish.
\newblock Finding fair and efficient allocations.
\newblock In {\em {EC}}, pages 557--574. {ACM}, 2018.

\bibitem{DBLP:journals/mst/BeiLMS21}
Xiaohui Bei, Xinhang Lu, Pasin Manurangsi, and Warut Suksompong.
\newblock The price of fairness for indivisible goods.
\newblock {\em Theory Comput. Syst.}, 65(7):1069--1093, 2021.
\newblock URL: \url{https://doi.org/10.1007/s00224-021-10039-8}, \href
  {http://dx.doi.org/10.1007/s00224-021-10039-8}
  {\path{doi:10.1007/s00224-021-10039-8}}.

\bibitem{BILO2022197}
Vittorio Bilò, Ioannis Caragiannis, Michele Flammini, Ayumi Igarashi,
  Gianpiero Monaco, Dominik Peters, Cosimo Vinci, and William~S. Zwicker.
\newblock Almost envy-free allocations with connected bundles.
\newblock {\em Games and Economic Behavior}, 131:197--221, 2022.
\newblock URL:
  \url{https://www.sciencedirect.com/science/article/pii/S0899825621001524},
  \href {http://dx.doi.org/https://doi.org/10.1016/j.geb.2021.11.006}
  {\path{doi:https://doi.org/10.1016/j.geb.2021.11.006}}.

\bibitem{conf/ijcai/BiswasB18}
Arpita Biswas and Siddharth Barman.
\newblock Fair division under cardinality constraints.
\newblock In {\em {IJCAI}}, 2018.

\bibitem{conf/aaai/BiswasB19}
Arpita Biswas and Siddharth Barman.
\newblock Matroid constrained fair allocation problem.
\newblock In {\em {AAAI}}, 2019.

\bibitem{DBLP:conf/ijcai/BouveretCEIP17}
Sylvain Bouveret, Katar{\'{\i}}na Cechl{\'{a}}rov{\'{a}}, Edith Elkind, Ayumi
  Igarashi, and Dominik Peters.
\newblock Fair division of a graph.
\newblock In Carles Sierra, editor, {\em Proceedings of the Twenty-Sixth
  International Joint Conference on Artificial Intelligence, {IJCAI} 2017,
  Melbourne, Australia, August 19-25, 2017}, pages 135--141. ijcai.org, 2017.
\newblock URL: \url{https://doi.org/10.24963/ijcai.2017/20}, \href
  {http://dx.doi.org/10.24963/ijcai.2017/20}
  {\path{doi:10.24963/ijcai.2017/20}}.

\bibitem{brams1995envy}
Steven~J Brams and Alan~D Taylor.
\newblock An envy-free cake division protocol.
\newblock {\em The American Mathematical Monthly}, 102(1):9--18, 1995.

\bibitem{journals/bqgt/Budish10}
Eric Budish.
\newblock The combinatorial assignment problem: Approximate competitive
  equilibrium from equal incomes.
\newblock {\em Journal of Political Economy}, 119(6):1061--1103, 2011.

\bibitem{DBLP:journals/mst/CaragiannisKKK12}
Ioannis Caragiannis, Christos Kaklamanis, Panagiotis Kanellopoulos, and Maria
  Kyropoulou.
\newblock The efficiency of fair division.
\newblock {\em Theory Comput. Syst.}, 50(4):589--610, 2012.
\newblock URL: \url{https://doi.org/10.1007/s00224-011-9359-y}, \href
  {http://dx.doi.org/10.1007/s00224-011-9359-y}
  {\path{doi:10.1007/s00224-011-9359-y}}.

\bibitem{DBLP:journals/teco/CaragiannisKMPS19}
Ioannis Caragiannis, David Kurokawa, Herv{\'{e}} Moulin, Ariel~D. Procaccia,
  Nisarg Shah, and Junxing Wang.
\newblock The unreasonable fairness of maximum nash welfare.
\newblock {\em {ACM} Trans. Economics and Comput.}, 7(3):12:1--12:32, 2019.

\bibitem{caragiannis2022little}
Ioannis Caragiannis, Evi Micha, and Nisarg Shah.
\newblock A little charity guarantees fair connected graph partitioning.
\newblock 2022.

\bibitem{conf/aaai/DrorFS21}
Amitay Dror, Michal Feldman, and Erel Segal{-}Halevi.
\newblock On fair division under heterogeneous matroid constraints.
\newblock In {\em {AAAI}}, pages 5312--5320, 2021.

\bibitem{dubinsHowCutCake1961}
L.~E. Dubins and E.~H. Spanier.
\newblock How to {{Cut A Cake Fairly}}.
\newblock {\em The American Mathematical Monthly}, 68(1):1--17, 1961.
\newblock \href {http://dx.doi.org/10.2307/2311357}
  {\path{doi:10.2307/2311357}}.

\bibitem{edward1999rental}
Francis Edward~Su.
\newblock Rental harmony: Sperner's lemma in fair division.
\newblock {\em The American mathematical monthly}, 106(10):930--942, 1999.

\bibitem{foley1966resource}
Duncan~Karl Foley.
\newblock {\em Resource allocation and the public sector}.
\newblock Yale University, 1966.

\bibitem{journals/ai/GargT21}
Jugal Garg and Setareh Taki.
\newblock An improved approximation algorithm for maximin shares.
\newblock {\em Artif. Intell.}, 300:103547, 2021.

\bibitem{hohneAllocatingContiguousBlocks2021}
Felix H{\"o}hne and Rob {van Stee}.
\newblock Allocating contiguous blocks of indivisible chores fairly.
\newblock {\em Information and Computation}, page 104739, March 2021.
\newblock \href {http://dx.doi.org/10.1016/j.ic.2021.104739}
  {\path{doi:10.1016/j.ic.2021.104739}}.

\bibitem{conf/sigecom/HuangL21}
Xin Huang and Pinyan Lu.
\newblock An algorithmic framework for approximating maximin share allocation
  of chores.
\newblock In {\em {EC}}, pages 630--631, 2021.

\bibitem{journals/corr/abs-2106-07300}
Halvard Hummel and Magnus~Lie Hetland.
\newblock Guaranteeing half-maximin shares under cardinality constraints.
\newblock {\em CoRR}, abs/2106.07300, 2021.

\bibitem{journals/aamas/HummelH22}
Halvard Hummel and Magnus~Lie Hetland.
\newblock Fair allocation of conflicting items.
\newblock {\em Auton. Agents Multi Agent Syst.}, 36(1):8, 2022.

\bibitem{journals/jacm/KurokawaPW18}
David Kurokawa, Ariel~D. Procaccia, and Junxing Wang.
\newblock Fair enough: Guaranteeing approximate maximin shares.
\newblock {\em J. {ACM}}, 65(2):8:1--8:27, 2018.

\bibitem{li2021fair}
Bo~Li, Minming Li, and Ruilong Zhang.
\newblock Fair scheduling for time-dependent resources.
\newblock {\em NeurIPS}, 2021.

\bibitem{DBLP:conf/sigecom/LiptonMMS04}
Richard~J. Lipton, Evangelos Markakis, Elchanan Mossel, and Amin Saberi.
\newblock On approximately fair allocations of indivisible goods.
\newblock In {\em {EC}}, pages 125--131. {ACM}, 2004.

\bibitem{DBLP:journals/siamdm/PlautR20}
Benjamin Plaut and Tim Roughgarden.
\newblock Almost envy-freeness with general valuations.
\newblock {\em {SIAM} J. Discret. Math.}, 34(2):1039--1068, 2020.

\bibitem{steinhaus1948problem}
Hugo Steinhaus.
\newblock The problem of fair division.
\newblock {\em Econometrica}, 16:101--104, 1948.

\bibitem{suksompongFairlyAllocatingContiguous2019}
Warut Suksompong.
\newblock Fairly allocating contiguous blocks of indivisible items.
\newblock {\em Discrete Applied Mathematics}, 260:227--236, May 2019.
\newblock \href {http://dx.doi.org/10.1016/j.dam.2019.01.036}
  {\path{doi:10.1016/j.dam.2019.01.036}}.

\bibitem{DBLP:journals/sigecom/Suksompong21}
Warut Suksompong.
\newblock Constraints in fair division.
\newblock {\em SIGecom Exch.}, 19(2):46--61, 2021.

\bibitem{sunEquitabilityWelfareMaximization2022}
Ankang Sun, Bo~Chen, and Xuan Doan.
\newblock {\em Equitability and {{Welfare Maximization}} for {{Allocating
  Indivisible Items}}}.
\newblock January 2022.

\bibitem{DBLP:conf/atal/SunCD21}
Ankang Sun, Bo~Chen, and Xuan~Vinh Doan.
\newblock Connections between fairness criteria and efficiency for allocating
  indivisible chores.
\newblock In Frank Dignum, Alessio Lomuscio, Ulle Endriss, and Ann Now{\'{e}},
  editors, {\em {AAMAS} '21: 20th International Conference on Autonomous Agents
  and Multiagent Systems, Virtual Event, United Kingdom, May 3-7, 2021}, pages
  1281--1289. {ACM}, 2021.
\newblock URL: \url{https://dl.acm.org/doi/10.5555/3463952.3464100}.

\bibitem{DBLP:journals/orl/Woeginger97}
Gerhard~J. Woeginger.
\newblock A polynomial-time approximation scheme for maximizing the minimum
  machine completion time.
\newblock {\em Oper. Res. Lett.}, 20(4):149--154, 1997.
\newblock URL: \url{https://doi.org/10.1016/S0167-6377(96)00055-7}, \href
  {http://dx.doi.org/10.1016/S0167-6377(96)00055-7}
  {\path{doi:10.1016/S0167-6377(96)00055-7}}.

\bibitem{conf/ijcai/00010G21}
Xiaowei Wu, Bo~Li, and Jiarui Gan.
\newblock Budget-feasible maximum nash social welfare is almost envy-free.
\newblock In {\em {IJCAI}}, pages 465--471, 2021.

\end{thebibliography}

\newpage

\appendix

\section*{Appendix}


\section{Missing Proofs in Section \ref{sec::pre}}
\begin{proof}[Proof of Lemma~\ref{lem:mms:computation}]
    Given a \emph{fair-chores} instance $\mathcal{I} = \langle N, E, \mathcal{V} \rangle$, it suffices to prove that one can answer whether $\MMS_i(\mathcal{I}) \geq q $ for $q \in \mathbb{Q}$ in polynomial time. If so, $\MMS_i(\mathcal{I})$ can be computed efficiently through binary search. We process the item from left to right using the following algorithm: 
    \begin{enumerate}
        \item If $n-1$ bundles have been removed, then terminate. Else, find the largest index $p$ with $v _ i (L(p) \cap E ) \ge q$;
        \item Remove $ L (p) \cap E $ and update $E \leftarrow E \setminus L(p)$. Return to Step 1.
    \end{enumerate}
    
This algorithm removes $n-1$ bundles $S_1, \ldots, S _ {n-1} $ from left to right in Step 2 and also leaves $S _ n $ (right-most) the set of remaining items at termination. In the following, we show that $ \min_{ j \in [n]} v _ i ( S _ j ) \ge q $ if and only if $  \MMS_i \ge q$. The ``only if'' direction, if $\min_{ j\in [n]} v _ i ( S _ j ) \geq q$, then due to the definition of MMS, we have $\MMS_i \geq \min_{ j\in [n]} v _ i ( S_ j ) \geq q$. For the ``if'' direction, it suffices to show that $\min_{ j\in [n]} v _ i ( S _ j ) < q $ implies that there does not exist a connected $n$-partition in which agent $i$'s value for every bundle is at least $q$. For the sake of contradiction, assume a connected $n$-partition $\mathbf{S} ^ {\prime}$ with $ \min_{ j \in [n]} v _ i ( S ^ {\prime} _ j ) \geq q$, and without loss of generality, $ S ^ {\prime} _ j $ is on the left of $S ^ {\prime} _ { j + 1 }$. According to Step 1, we have $ S _ 1 ^ {\prime}  \subseteq S _ 1$, which then implies $ S ^ { \prime} _ 1 \cup S ^ {\prime} _ 2 \subseteq S _ 1 \cup S _ 2 $; otherwise, contradicts to the rule of Step 1 when removing bundle $S _ 2$. By induction, we have $ \cup _{ j = 1} ^ {n-1} S _ j ^ {\prime} \subseteq \cup _{ j = 1} ^ {n-1} S _ j $ and accordingly, $S _ n \subseteq S ^ {\prime} _ n $. Thus, $ v _ i ( S _ j ) \geq q $ holds for all $ j \in [n]$, contradiction. Therefore, it holds that $ \min_{ j \in [n]} v _ i ( S _ j ) \ge q $ if and only if $  \MMS_i \ge q$. The running time of the aforementioned algorithm is $O(m)$.
\end{proof}

\section{Missing Materials in Section \ref{sec:goods}}

\subsection{Missing Algorithms}

	\begin{algorithm}[H]
	\caption{\hspace{-2pt}{$\Matching'(\cI,\alpha)$}}
	\label{ALG:matching:cardinality}
	\begin{algorithmic}[1]
		\REQUIRE Instance $\mathcal{I} = \langle N, E, \mathcal{V} \rangle$, a parameter vector $\calpha = (\alpha_i)_{i\in N}$ where $\alpha_i \ge 0$.
		\ENSURE A partial allocation $\mathbf{A}'$ where each agent is allocated at most one item.
		\STATE Construct a weighted bipartite graph $G = (N\cup E, N \times E )$ where agents are vertices on one side and items are vertices on the other side. 
		For each $i\in N$ and $j\in E$, there is an edge $(i,j)$ with weight $v_i(e_j)$ if $v_i(e_j) \ge \alpha_i$. 
		
		\STATE Compute a maximum {\em cardinality} matching $\mu$ of $G$ and denote by $\mu(i)$ the item matched to agent $ i $. If agent $ i $ is unmatched, let $\mu ( i ) =\emptyset$. Construct the partial allocation $\mathbf{A}'$ with $A' _ i = \{\mu (i)\}$ for every $ i \in N $ such that $A'_i\neq \emptyset$.
	
		\RETURN $\mathbf{A}'$
	\end{algorithmic}
\end{algorithm}

\subsection{Price of MMS for Indivisible Goods}
\medskip
\noindent\begin{proof}[Proof of Lemma \ref{lemma::MMS-goods-U-add-2}]
    According to the Subroutine $\Matching(\mathcal{I}, \calpha)$ in Step 1 of Algorithm~\ref{ALG::MMS-good-utilitarian}, it holds that $ v_ i ( A _ i ) \geq v _ i ( A ^ {\prime} _ i ) \geq \max \{ \MMS_i(\mathcal{I}), \frac{1}{4n} \}$ for each $ i \in N _ 0 $. Then, the remaining is to prove the statement holds for every $ i \in N\setminus N _ 0 $, which is equivalent to prove that each agent $ i \in N\setminus N _ 0 $ can receive a bundle from $\MovingKnife(\mathcal{I}, \mathbf{A}^ {\prime}, \calpha)$ in Step 5 of Algorithm~\ref{ALG::MMS-good-utilitarian}.
    
    We now fix $ i \in N\setminus N _ 0$. If $\MMS_i \geq \frac{1}{4n}$, then by arguments similar to the proof of Lemma~\ref{lemma::MMS-good-U-add-1}, one can verify that agent $ i $ must receive a connected bundle from Step 5 of Algorithm~\ref{ALG::MMS-good-utilitarian} so that her value is $ v _ i ( A _ i ) \geq \max \{ \MMS_i, \frac{1}{4n} \}$. We can further focus on the case of $\MMS_i < \frac{1}{4n}$ and assume for the sake of contradiction, agent $ i $ does not receive a bundle in Step 5. Since agent $ i $ is not matched in $\Matching(\mathcal{I}, \calpha)$, we can claim $ v_ i (e) < \frac{1}{4n}$ for every $ e \in E ^ {\prime}$. Denote by $\tilde{N}$ the set of agents who receive bundles in Step 5. We can assume $ |\tilde{N}| \leq n - | N _ 0 | - 1 $; otherwise, the statement already holds. For each $ j\in \tilde{N}$, let $ B _ j $ be the bundle assigned to agent $ j $ in Step 5 and $ e _ {j _ R} \in B _ j $ the right-most item of $ B _ j $. Due to $ v_ i (e_{j_R}) < \frac{1}{4n}$ and Subroutine $\MovingKnife$ in Step 5, we have $ v _ i ( B _ j ) = v _ i ( B _ j \setminus \{ e _ {j _ R} \}) + v _ i ( e _ {j_R}) < \frac{1}{2n}$. We now upper-bound the value of agent $ i $ on assigned items. As in Step 1, $\Matching(\mathcal{I}, \calpha)$ computes the maximum weighted matching, we have $ v _ i ( E _ 0 ) \leq \UW(\mathbf{A} ^ {\prime}) < \frac{1}{4}$. Then, agent $ i $'s value is at most
    $$
    v _ i ( E _ 0 ) +  v _ i ( \cup _ { j \in \tilde{N}} B _ j )< \frac{1}{4} + \frac{n - | N _ 0 | - 1}{ 2n }.
    $$
    As a consequence, agent $ i $ has value at least $\frac{n + 2 | N _ 0 | + 2 }{ 4n}$ on the unassigned items due to normalized valuations. Since at most $ n -1 $ agents receiving a bundle, the set of unassigned items is then composed by at most $ n $ contiguous blocks. Thus, due to the pigeonhole principle, agent $ i $ has value at least $\frac{n + 2 | N _ 0 | + 2 }{ 4n ^ 2} > \frac{1}{4n}$ on one of the unassigned contiguous blocks, based on which agent $ i $ should receive a bundle in Step 5, contradiction.
\end{proof}

\medskip
\noindent\begin{proof}[Proof of Lemma \ref{lem:goods:u:mms:lb}]
	Consider an instance with $n$ agents and a set $ E = \{ e_1, \ldots, e _{ 2n }\}$ goods. For $i = 1, \ldots, \sqrt{n}$, the valuation function $ v _ i ( \cdot)$ is: $ v _ i ( e _ j ) = \frac{1}{2\sqrt{n}}$ for $ 2(i-1)\sqrt{n} + 1 \leq j \leq 2i \sqrt{n}$ and $ v _ i ( e _ j ) = 0 $ for other $j$. For $ i \geq \sqrt{n} + 1$, the valuation functions is: $ v _ i ( e _ j ) = \frac{1}{2n}$ for any $ j \in [2n]$. One can compute $\MMS_i = 0$ for $ i \leq \sqrt{n}$ and $\MMS _ i = \frac{1}{n}$ for $ i \geq \sqrt{n} + 1$. In a utilitarian welfare-maximizing allocation $\mathbf{O}$, each agent $ i \leq \sqrt{n}$ receives all goods on which she has positive value and agent $ i > \sqrt{n}$ receives none. We can compute $\OPT_U = \UW(\mathbf{O}) = \sqrt{n}$. But for agent $i \geq \sqrt{n} + 1 $, since $ v _ i(O _ i ) = 0 < \MMS _ i $, she violates MMS under allocation $\mathbf{O}$. To make such an agent satisfy MMS fairness, two adjacent goods must be assigned to her. Thus, in total $2(n - \sqrt{n})$ goods need to be assigned to the latter $ n - \sqrt{n}$ agents, which makes only $2\sqrt{n}$ goods can be assigned to the first $\sqrt{n}$ agents in an MMS allocation. Consequently, for an arbitrary MMS allocation $\mathbf{A}$, we have $\UW(\mathbf{A}) \leq 2 - \frac{1}{\sqrt{n}}$ and therefore,
	$$
	\PoF \geq \frac{\sqrt{n}}{ 2 - \frac{1}{\sqrt{n}}}=\Omega(\sqrt{n}),
	$$ which completes the proof of the lemma.
\end{proof}

\medskip
\noindent\begin{proof}[Proof of Lemma \ref{lem:goods:e:mms:ub}]
We first prove the following claim. 

\begin{claim}\label{lemma::goods-RW-MMS-add1}
	In Subroutine $\Matching^{\prime}(\mathcal{I}, \calpha)$, if agent $ i $'s vertex has degree at least one in graph $G$, then she is matched by $\mu$.
\end{claim}
\begin{proof}
	Denote by $\bar{N}$ the set of agents with degree at least one in $G$. Clearly, $\MMS_i < \frac{1}{2n}\OPT_E$ holds for every $i \in \bar{N}$. For the sake of contradiction, assume the matching $\mu$ is not $\bar{N}$-perfect. Then according to the Hall's theorem, there exists a subset $ N ^* \subseteq \bar{N}$ satisfying $|N^*| > | D _{{G}} (N^*)|$ where $D _{{G}} (N^*)$ is the neighbourhood of $N^*$ in ${G}$. 
	
	We then focus on the set $E \setminus D _{{G}} (N^*)$ and claim that no connected subset $ P \subseteq E \setminus D _{\bar{G}} (N^*)$ is able to bring value $ v _ t (P) \geq \OPT_E$ for some agent $ t \in N^*$. Suppose not, and assume that agent $ j \in N^*$ has value $ v _ j (P^*) \geq \OPT_E$ where $ P ^* \subseteq E \setminus D _{{G}} (N^*)$ is connected. Due to the choice of $\calpha$, we have $ v _ j ( e ) < \frac{1}{2n}\OPT_E$ for each $e \in P^*$. Thus, set $P^*$ is able to be partitioned into $n$ connected subsets $\{ P _ l ^*\}_{l=1} ^ n$ such that $ v _ j (P_l^*) \geq \frac{1}{2n}\OPT_E$ for all $ l $, which then leads to $\MMS_j \geq \frac{1}{2n}\OPT_E$, contradiction. Thus, given an agent $ i \in N^*$ and a connected bundle $S\in \mathcal{S}$, $ v _ i (S) \geq \OPT_E$ if and only if $S \cap D _{{G}} (N^*) \neq \emptyset$. Notice that $ |N^*| > |D _{{G}} (N^*) | $, then it is impossible to make every agent in $N^*$ receive value at least $\OPT_E$, contradiction. Therefore, matching $\mu$ must be $\bar{N}$-perfect.
\end{proof}

Claim~\ref{lemma::goods-RW-MMS-add1} implies that if agent $i$ with $\MMS_i<\frac{1}{2n}\OPT_E$ values a single item $v _i(e) \ge \frac{1}{2n}\OPT_E$, then she will be matched to a single item in Step 1 of Algorithm~\ref{ALG::MMS-good-egalitarian} and has value 
$$v _ i ( A_i) \geq v _ i (\mu(i)) \ge \frac{1}{2n}\OPT_E \geq \MMS_i, $$ where $\mu$ is the maximum cardinality matching in Subroutine $\Matching^{\prime}(\mathcal{I}, \calpha)$. Thus, if $N_0 = N$, the statement is proved, and we can further assume $N _ 0 \subsetneq
	 N $. Fix agent $ i \in N \setminus N _ 0$ and split the proof into two cases.


	\emph{Case 1}: $\MMS _ i \geq \frac{1}{2n}\OPT_E$. The proof of this case is similar to the proof of Lemma~\ref{lemma::MMS-good-U-add-1} and we omit it.

	\emph{Case 2}: $\MMS _ i < \frac{1}{2n}\OPT_E$. It suffices to show that agent $ i $ receives a bundle in Step 2 of Algorithm~\ref{ALG::MMS-good-egalitarian}. Based on Claim~\ref{lemma::goods-RW-MMS-add1}, agent $ i $ has value $ v _ i ( e ) < \frac{1}{2n}\OPT_E$ for each $ e \in E$, and thus, $ v _ i ( \mu(j) ) < \frac{1}{2n}\OPT_E$ holds for $j\in N_0$ where $\mu(j)$ is the matching computed in Subroutine $\Matching^{\prime}$. For each $j \in N\setminus N _ 0 $ and $ j \neq i $, denote by $B _ j $ (if exists) the bundle received by agent $j$ in Step 2 of Algorithm~\ref{ALG::MMS-good-egalitarian}, and accordingly, $ v _ i ( B _ j ) < \frac{1}{n}\OPT_E$ holds; otherwise, agent $ i $ will receive a bundle at the time when agent $ j $ is picked in Subroutine $\MovingKnife$ in Step 2 of Algorithm~\ref{ALG::MMS-good-egalitarian}. Therefore, even after all other $n-1$ agents receiving connected bundles in either Step 1 or Step 2, agent $ i $ still has value no less than $ 1- \OPT_E>1- \frac{1}{n}$ on the unassigned items that are composed by at most $n$ connected subsets. By pigeonhole principle, there exists a subset with value at least $\frac{n-1}{n^2}>\MMS_i$ for agent $i$. Hence, agent $ i $ receives a connected bundle $B _ i $ in Subroutine $\MovingKnife$ in Step 2 of Algorithm~\ref{ALG::MMS-good-egalitarian} and has value $$ v _ i ( A _ i) \geq v _ i ( B _ i ) \geq \max\{ \frac{1}{2n}\OPT_E, \MMS_i \} .$$
	
	Therefore, we can conclude that Algorithm~\ref{ALG::MMS-good-egalitarian} can output the connected allocation $\mathbf{A}$ in which each agent $i$ receives value $v _ i (A _ i ) \geq \max\{ \frac{1}{2n}\OPT_E, \MMS_i \}$.
\end{proof}

\subsection{Price of PROP1 for Indivisible Goods}

\begin{proof}[Proof of Lemma \ref{lemma::PROP1-good-UW-add-1}]
	Notice that if an agent $ i $ has valuation $ v _ i ( e _ j ) \geq \frac{1}{n}$ on good $ e _ j $, then she satisfies PROP1 when she receives empty bundle. Accordingly, we can further focus on the instance $ \mathcal{I}$ in which $ v _ i (e _ j ) < \frac{1}{n}$ holds for any $ i, j $. We then construct a corresponding cake-cutting instance $ \mathcal{I} ^ {\prime}$ with $n$ agents and the cake being the interval $[0,m]$ where $m = | E |$. In $\mathcal{I}^{\prime}$, the value of interval $[a,b]$ for each agent $i$ is equal to $\int_{a}^b f_i(x)dx$ where $f_i$ is agent $i$'s density function. Each agent $ i $ has a piecewise constant density function $f _ i ( x ) = v _ i ( e _ j )$ on interval $ [j-1,j ]$, and thus, agent $i$ has value $ v _ i ( e _ j )$ on piece $[j-1, j ]$. 
	According to \cite{dubinsHowCutCake1961}, instance $ \mathcal{I} ^ {\prime}$ admits a connected proportional allocation $\pi$, in which w.l.o.g, agents $1, \ldots, n$ receive the 1st, 2nd, ..., $n$-th piece of cake from left to right, and each agent $ i $ receives interval $[\pi _ { i - 1 }, \pi _ i ]$. We then transfer $\pi$ into allocations of $ \mathcal{I} $ that satisfy PROP1 and have an absolute welfare guarantee.
	
	Since each agent $ i $ has value at least $\frac{1}{n}$ on her piece $[\pi_{i-1}, \pi_i]$, we claim that $[\pi_{i-1}, \pi _ i ] \notin [ j-1, j]$ for any pair of $ i ,j $ due to the property of $\mathcal{I}$. As a result, for each $ j \in [m]$, interval $[j-1, j]$ is either covered by a interval $[\pi_{ p - 1}, \pi _ p ]$ or not covered by a single interval but intersects with two connected pieces received by agents in $\pi$. 
	We then construct two connected allocations of $\mathcal{I}$. In allocation $\mathbf{A} ^ L $ (resp. $\mathbf{A} ^ R $), each good $e _ j $ is assigned to agent $ i $ if $[j-1, j] \in [\pi _ {i-1}, \pi _ i ]$, and $e _ j$ is assigned to agent $ p $ (resp. $q$) if $[j-1, j]$ intersects\footnote{If the intersection of two intervals is a single point, then we regard their intersection as an empty set.} with two connected pieces $[\pi_{p-1}, \pi _ p]$, $[\pi_{q-1}, \pi _ q ]$ with $p< q$. In the following, we first show both $\mathbf{A} ^L$ and $\mathbf{A} ^ R$ are connected PROP1 allocations and then prove that one of them has utilitarian welfare at least $\frac{1}{2}$.
	
	The connectivity of $\mathbf{A} ^ L$ and $\mathbf{A} ^ R $ comes from the connectivity of $\pi$. We then prove the PROP1 of allocation $\mathbf{A} ^ L$ and fix an agent $ i $ who receives the piece of cake $[\pi _ { i -1}, \pi _ i ]$. If $\pi _ { i - 1 } \in \mathbb{N}^+$, then the bundle received by agent $ i $ is $ A _ i ^ L = \{ e _ j \mid \pi_{ i - 1} \leq j \leq\lceil \pi_i \rceil \}$. Accordingly, we have the following
	$$
	v _ i ( A _ i ^ L ) = \int_{\pi_ { i - 1 }}^{ \lceil \pi_i  \rceil} f _ i ( x ) dx \geq \int_{\pi_{ i - 1}}^{ \pi _ i } f _ i ( x )dx \geq \frac{1}{n},
	$$ where the last inequality is due to the proportionality of $\pi$. If $\pi _ { i - 1 } \notin \mathbb{N}^+$, agent $ i $ receives $ A ^ L _ i = \{ e _j \mid \lceil \pi_{ i - 1} \rceil \leq j \leq \lceil \pi_{ i} \rceil \}$. Notice that item $ A ^L _i \cup  \{e _ {\lfloor \pi_{ i - 1} \rfloor}\}\in \mathcal{S}$ and we have
	$$
	v _ i ( A ^ L _ i \cup \{ e _ {\lfloor \pi_{ i - 1} \rfloor} \}) = \int_{ \lfloor \pi_{i-1} \rfloor} ^ {\lceil \pi _ i  \rceil } f _ i ( x )dx \geq \int_{\pi_{ i - 1}}^{ \pi _ i } f _ i ( x )dx \geq \frac{1}{n},
	$$ which then implies that agent $ i $ also satisfies PROP1. Therefore, we can conclude that $\mathbf{A} ^ L$ satisfies PROP1. By a similar argument, one can prove that $\mathbf{A} ^ R $ is also a connected PROP1 allocation.
	
	As each item $ e _ j $ of $ \mathcal{I} $ corresponds to an interval $ [j-1, j]$ of $ \mathcal{I} ^ {\prime}$, then any connected subsets $ S \subseteq E $ also corresponds to a subinterval of $[0,m]$. For each agent $ i $, denote by $[x _{ i -1} ^ L, x _ i ^ L ]$ and $ [x _ { i - 1} ^ R, x _ i ^ R ]$ the corresponding intervals of $ A ^ L_ i $ and $ A ^ R _ i $, respectively. By the construction $\mathbf{A} ^ L$ and $\mathbf{A} ^ R $, for each $ i \in [n]$, we have $\min\{ x ^ R _ { i -1}, x ^ L_{ i -1 } \} \leq \pi_{ i - 1 }$ and $ \max\{ x ^ R _ i, x ^ L _ i  \} \geq \pi _ i $. Then, for each $i\in [n]$, we have the following inequality
	$$
	\int_{x ^ L _{ i - 1 }}^{ x ^ L _ i } f _ i ( x )d x + \int_{ x _ { i -1 } ^ R }^{ x _ i ^ R }f _ i  ( x ) dx \geq \int_{\pi_{ i - 1}}^{ \pi_i } f _ i ( x ) dx,
	$$ which then implies
	$$
	\sum_{t=1} ^ {n} \int_{x ^ L _ {t-1}}^{ x ^ L _{ t}} f _ t (x)dx + \sum_{t=1} ^ {n} \int_{x ^ R _ {t-1}}^{ x ^ R _{ t}} f _ t (x)dx \geq \sum_{t=1} ^ {n}\int_{\pi_{t-1}}^{\pi_t} f _ t (x)dx.
	$$ The right hand side of the last inequality is the utilitarian welfare of allocation $\pi$ and should be at least one because $\pi$ is a proportional allocation. Furthermore, the left hand side is actually equals to $\UW(\mathbf{A} ^ L) + \UW(\mathbf{A} ^ R)$ due to the construction of the density function. Consequently, we can conclude one of $\mathbf{A} ^ L$ and $ \mathbf{A} ^ R$ has utilitarian welfare at least $\frac{1}{2}$.
	
	The above-mentioned contiguous proportional cake-cutting solution $\pi$ can be found efficiently. Then, this construction proof can be easily transferred to an efficient algorithm on computing the connected PROP1 allocation with utilitarian welfare at least $ \frac{1}{2}$.
\end{proof}

\medskip
\noindent\begin{proof}[Proof of Lemma \ref{lem:goods:u:prop1:lb}]
Consider an instance with $n$ agents and a set $ E = \{ e _ 1, \ldots, e _ {n+1}\}$ goods. For $ i =1, \ldots, \sqrt{n}$, the valuation function $ v _ i (\cdot)$ is: $ v _ i ( e _ j ) = \frac{1}{\sqrt{n}}$ for $( i -1)\sqrt{n} + 1 \leq j \leq i\sqrt{n}$ and $ v _ i ( e _ j ) = 0 $ for other $j$. For $i\geq \sqrt{n} + 1$, the function $ v _ i (\cdot)$ is: $ v _ i ( e _ j ) = \frac{1}{n + 1 }$ for $ j \in [n+1]$. Consider allocation $\mathbf{O}$ with $O_i = \{ e_{(i-1)\sqrt{n}+1},\ldots, e_{i\sqrt{n}}\} $ for each $i \leq \sqrt{n}$, $O_{n} = \{e_{n+1}\}$ and $O_j = \emptyset$ for other $j$. One can verify that $\OPT_U \geq \UW(\mathbf{O}) = \sqrt{n} + \frac{1}{n + 1 }$. However, agent $ j = n -1 $ violates PROP1 under allocation $\mathbf{O}$. In a connected PROP1 allocation $\mathbf{A}$, each agent $j > \sqrt{n}$ receives at least one item. As a consequence, we have $ \UW(\mathbf{A}) \leq 2 - \frac{\sqrt{n}-1}{\sqrt{n}(n+1)}$, which then implies
$$
\PoF  \geq \frac{\sqrt{n} + \frac{1}{n + 1 }}{2 - \frac{\sqrt{n}-1}{\sqrt{n}(n+1)}}  \geq \frac{\sqrt{n} + \frac{1}{n + 1 }}{2} = \Omega(\sqrt{n}),
$$ which completes the proof.
\end{proof}

\section{Missing Proofs in Section \ref{sec:chores}}

\subsection{Price of MMS for Indivisible chores}

\begin{proof}[Proof of Lemma \ref{Lemma::MMS-UW-chore}]
	With known $\MMS$ values, Algorithm \ref{alg:chores:MMS} allocates all items in $O(mn)$ time since the number of agents is reduced by one in each iteration of while-loop, and all remaining items are assigned to the last agent.
	By Lemma \ref{lem:mms:computation}, the computation of $\MMS_i$'s can be done in polynomial time. 
	
    Next we prove the returned allocation $\cA$ is MMS fair.
    Renumber the agents from 1 to $n$ by the order they receive bundles in the algorithm, where agent 1 is the first to receive a bundle and agent $n$ is the last. 
	For $i \in [n-1]$, agent $i$ receives chores in either Step \ref{step:chores:mms:1} or \ref{step:chores:mms:2}. 
	If $A _ i $ is assigned in Step \ref{step:chores:mms:1}, we have $ v _ i ( A _ i ) \geq \max\{ \MMS _ i , \frac{2}{n}\} \geq \MMS _ i $. For the latter, since $ |A _ i | = 1 $, we clearly have $ v _ i ( A _ i ) \geq \MMS _ i $. 
	Thus, MMS fairness is satisfied by agents $[n-1]$.
	
	It remains to show $v_n(A_n) \ge \MMS_n$, and we consider two cases.
	
	\medskip
	
	\emph{Case 1:} $\MMS _ n \ge -\frac{2}{n} $. 
	Let $\mathbf{S} = ( S _ 1, \ldots, S _ n )$ be a connected $\MMS _ n $-defining partition and bundle $ S _ i $ is on the left of $ S _ j $ for any $ i < j $. 
	For allocation $\mathbf{A}$, we have also $ A _ i $ is on the left of $ A _ j $ for any $ i < j $. 
	Then we have the following claim.
	\begin{claim}\label{claim:chores:alg:1}
	For any $1\le k\le n$, $\cup_{j=1}^k S_j \subseteq \cup_{j=1}^k A_j$.
	\end{claim}
	\begin{proof}[Proof of Claim \ref{claim:chores:alg:1}]
	When $k=1$, we have $ S _ 1 \subseteq A _ 1 $; otherwise $ | S _ 1 | > | A _ 1 |$ and $ v _ n ( S _ 1 ) \geq \MMS _ n$, which means agent $ n $  should be allocated $S_1$ at the time when $A_1$ is assigned to agent 1 and thus leads to a contradiction. 
	Similarly, $ S _ 1 \cup S _ 2 \subseteq A _ 1 \cup A _ 2 $; otherwise, since $ S _ 1 \subseteq A _ 1 $, we have $ | S _ 2 \setminus A _ 1 | > | A _ 2 |$ and $v _ n ( S _ 2 \setminus A _ 1 ) \geq v _ n ( S _ 2 ) \geq \MMS_ n $, which contradicts to the choice of agent 2 in while-loop. 
	By induction, we have the claim.
\end{proof}
	
	Claim \ref{claim:chores:alg:1} implies that $ A _ n \subseteq S _ n$, and thus $ v _ n ( A _ n ) \geq v _ n ( S_  n ) \geq \MMS _ n $, which completes the proof of this case.
	
	\medskip
	
	\emph{Case 2:} $\MMS_n < -\frac{2}{n}$.
	For any $ i \in [n-1]$, we have three sub cases described:
	\begin{enumerate}
	\item \emph{Case 2.1:} bundle $ A _ i $ is assigned in Step \ref{step:chores:mms:1} and $ v _ n ( A _ i ) \geq -\frac{2}{n}$.
	\item \emph{Case 2.2:} bundle $ A _ i $ is assigned in tep \ref{step:chores:mms:1} and $ v _ n ( A _ i ) < - \frac{2}{n}$.
	\item \emph{Case 2.3:} bundle $ A _ i $ is assigned in Step \ref{step:chores:mms:2}.
	\end{enumerate}
	If Case 2.3 happens, agent $n$'s value on $A _ i $ satisfies $ v _ n ( A _ i) < - \frac{2}{n}$ due to the condition of Step 4. Thus, for both Case 2.2 and 2.3, we have $ v _ n ( A _ i ) < -\frac{2}{n}$. 
	
	We then focus on Case 2.1. 
	Given an agent $i$ satisfying the condition of Case 2.1, let $\bar{e}\notin A_i$ be the item right-connected with bundle $ A _ i $. 
	Then, we must have $ v _ n ( A _ i \cup \{ \bar{e}\}) < -\frac{2}{n}$; otherwise, at the time when $A _ i $ is assigned, Algorithm \ref{alg:chores:MMS} should choose agent $n$ instead.
	Denote by $\mathcal{P}$ the set of agents who satisfy the condition of Case 2.1 and $\mathcal{Q}$ the set of agents whose bundle is right-connected to some $A _ i $ with $ i \in \mathcal{P}$. 
	Note that $|\mathcal{P}| = | \mathcal{Q}| $ and $\mathcal{P} \cap \mathcal{Q}$ can be non-empty. 
	For any $ i \in \mathcal{P}$, it uniquely maps a $ j _ i \in \mathcal{Q}$ where $ A _{ j _ i }$ is right-connected to $ A _ i $. 
	The value of agent $n$ on this union bundle satisfies $ v _ n ( A _ i \cup A _{ j _ i }) < - \frac{2}{n}$.
	We remark that the collection $\{ A _ i \} _ { i =1} ^ {n}$ can be classified by whether $ i \in \mathcal{P} \cup \mathcal{Q}$ or not. 
	For these $ i \notin \mathcal{P} \cup \mathcal{Q}$, we have $ v_ n ( A _ i ) < -\frac{2}{n}$. 
	Then we have the following claim.
	
	\begin{claim}\label{claim:chores:alg:2}
	$ v _ n ( \cup_{ i \in \mathcal{P} \cup \mathcal{Q}} A _ i ) < -\frac{|\mathcal{P} \cup \mathcal{Q}|}{n}$.
	\end{claim} 
	
	\begin{proof}[Proof of Claim \ref{claim:chores:alg:2}]
    Note that $\cup_{ i \in \mathcal{P} \cup\mathcal{Q}} A _ i $ contains several maximal connected blocks $\mathcal{A} _ s$'s on $E$, and each maximal block $\mathcal{A} _ s$ is the union of at least two bundles in $\{ A _i  \} _ { i \in \mathcal{P} \cup \mathcal{Q}}$.
	In particular, $\mathcal{A} _ s $ is in the form: the left most bundle $A _ p $ satisfies $ p \in \mathcal{P}\setminus \mathcal{Q}$; the right most bundle $ A _ q $ satisfies $ q \in \mathcal{Q}\setminus \mathcal{P}$; other bundles $A _ r $ satisfies $ r \in \mathcal{P} \cap \mathcal{Q}$. 
	Let $k_s$ be the number of blocks of $\{ A _ i \}_{i\in \mathcal{P} \cup \mathcal{Q}}$ in $ \mathcal{A} _ s$.
	Recall that for $i \in \mathcal{P}$ and $ j _ i $ with $ A _ {j_i}$ right-connected to $ A _i $, we have $ v _ n ( A _ i \cup A _ { j _ i }) < - \frac{2}{n}$. 
	Thus, if $k_s$ is even, 
	$$ 
	v _ n (\mathcal{A} _ s ) < - \frac{k_s}{2} \times \frac{2}{n}  = -\frac{k_s}{n};
	$$ 
	If $k_s$ is odd, since the right most bundle $A_q$ is not in $\{ A _ i \} _{ i \in \mathcal{P}}$, i.e., $v_n(A_q) < -\frac{2}{n}$, then 
	$$
	v _ n ( \mathcal{ A } _ s) < -\frac{k_s-1}{2}\times\frac{2}{n} - \frac{2}{n} = -\frac{k_s+1}{n}.
	$$
	Therefore, we have $ v _ n (\mathcal{A} _ s ) < -\frac{k _ s}{n}$ for all $\mathcal{A} _ s$. 
	Summing over all the maximal connected blocks $\mathcal{A} _ s$, we have
	$$
	 v _ n (\cup_{ i \in \mathcal{P} \cup \mathcal{Q} } A _ i) 
	 = \sum_{s} v _ n ( \mathcal{A} _ s) < - \sum_{s} \frac{k_s}{n} 
	 = - \frac{| \mathcal{P} \cup \mathcal{Q}|}{n}.
	$$
\end{proof}
	

	Combing Claim \ref{claim:chores:alg:2} and the fact that $ v_ n ( A _ i ) < -\frac{2}{n}$ for $ i \notin \mathcal{P} \cup \mathcal{Q}$, we have
	$$
	\begin{aligned}
	 v _ n ( \cup_{ i =1} ^ {n-1} A _ i )
	 & = \sum _ { i \in \mathcal{P}\cup \mathcal{Q}} v _ n ( A _ i) + \sum_{i \notin \mathcal{P}\cup \mathcal{Q}} v _ n ( A _ i) \\
	 & < - \frac{|\mathcal{P} \cup \mathcal{Q}|}{n} - \frac{2(                           n-1 - | \mathcal{P}\cup\mathcal{Q}|)}{n} \\
	 & \leq - \frac{n-1}{n},
	 	\end{aligned}
	$$ where the last inequality is due to $ | \mathcal{P} \cup \mathcal{Q}| \leq n - 1 $. 
	Since the valuations are normalized, we have $ v _ n ( A _n ) \geq -\frac{1}{n} \geq \MMS _ n $.
	
	In conclusion, $ v _ i ( A _ i ) \geq \textnormal{MMS} _ i $ for all $i \in N$, which finishes the proof of Lemma \ref{Lemma::MMS-UW-chore}.
\end{proof}

\medskip
    
\noindent\begin{proof}[Proof of Theorem \ref{thm::MMS-chore-UW}]
    We only prove for $n \ge 3$ in this section and defer the discussion of $n=2$ to Section \ref{sec:n=2}. 
    The proof is split into the following two lemmas.
    
    \begin{lemma}
\label{lem:chores:MMS:mid$Utilitarian$:n>2:upper}
    For the problem of Chores $\mid$ Utilitarian $\mid$ MMS with $n\ge 3$, $\PoF \le 3n$.
\end{lemma}

\begin{lemma}
\label{lem:chores:MMS:mid$Utilitarian$:n>2:lower}
    For the problem of Chores $\mid$ Utilitarian $\mid$ MMS with $n \ge 3$, $\PoF \ge \frac{n}{5}$.
\end{lemma}



\begin{proof}[Proof of Lemma \ref{lem:chores:MMS:mid$Utilitarian$:n>2:upper}]
	For the upper bound, if the maximum utilitarian welfare is at least $-\frac{1}{n}$, then the allocation that maximize the utilitarian welfare must be MMS. 
	Thus it suffices to consider the case when the optimal utilitarian welfare is strictly less than $ - \frac{1}{n}$. 
	Let $\mathbf{A}$ be the allocation returned by Algorithm \ref{alg:chores:MMS}. 
	By Lemma~\ref{Lemma::MMS-UW-chore}, allocation $\mathbf{A}$ is MMS. 
	Let $N_1$ and $N_2$ be the sets of agents whose bundles are assigned in Step \ref{step:chores:mms:1} or \ref{step:chores:mms:2} respectively, and for any $ i \in N _ 1 $, $ v _ i ( A _ i ) \geq -\frac{2}{n} $. 
	Without loss of generality, assume agent $n$ is the last one to receive a bundle. 
	Note $n $ may or may not be in $ N _ 1 \cup N_ 2$, and by the proof of Lemma~\ref{Lemma::MMS-UW-chore}, $ v _ n ( A _ n ) \geq - \frac{1}{n}$. 
	
	Next we bound the welfare of agents in $N _ 2 $.
	Denote by $ N _ 2 = \{ i_ 1, i _ 2, \ldots, i _ p \}$, and bundle $ A _ { i _ l }$ is on the left of $ A _ { i _ k }$ for any $ l < k \leq p$. 
	By the selection of agents in Step \ref{step:chores:mms:2:con}, $ v _ { i _ p } ( A _ { i _ k }) \leq v _ { i _ k }( A _ {i _ k })$  for $k \leq p$. 
	Thus, the welfare of agents in $N _ 2 $ is
	$$
	\sum_{ i \in N _ 2 } v _ i ( A _ i) = \sum_{ k \in [p]} v _ {i_k }( A _ {i _ k }) \geq \sum_{ k \in [p]} v _ {i _ p } ( A _ {i_k}) \geq - 1,
	$$ 
	where the last inequality is by the normalized valuations. 
	Accordingly, the overall welfare of $\mathbf{A}$ is bounded by
	$$
	\begin{aligned}
	\sum_{ i \in [n]} v _ i ( A _ i ) &\geq \sum_{ i \in N_1} v _ i ( A _ i) + \sum_{ i \in N _ 2} v _ i ( A _ i ) + v _ n ( A _ n ) \\
	& \geq -\frac{2}{n}\cdot |N _ 1 | - 1 - \frac{1}{n}\\
	& \geq -3,
	\end{aligned}
	$$ 
	where the last inequality is due to $ | N _ 1 | \leq n - 1 $ and the condition in Step \ref{step:chores:mms:1:con}. 
	Thus, we find an MMS allocation with welfare at least $ - 3$, and the price of MMS is at most $3n$.
	\end{proof}

	\medskip
	
	\begin{proof}[Proof of Lemma \ref{lem:chores:MMS:mid$Utilitarian$:n>2:lower}]
	As for the lower bound, consider the instance with $n$ agents and a set $E = \{ e _ 1, \ldots, e _ {3n - 2} \}$ of $3n-2$ chores. 
	The valuations are shown in Table \ref{tab:chore-utilitarian-mms-n>3}.
	\begin{table}[h]
	    \centering
	    \begin{tabular}{c|c|c|c|c|c|c}
	         Items & $e_1$ & $\cdots$ & $e_{2n}$ & $e_{2n+1}$ & $\cdots$ & $e_{3n-2}$  \\
	         \hline
	         $v_1(\cdot)$ & $-\frac{1}{n^2}$ & $\cdots$ & $-\frac{1}{n^2}$ & $-\frac{1}{n}$ & $\cdots$ & $-\frac{1}{n}$  \\
	         $v_i(\cdot)$ for $i\ge 2$ & $-\frac{1}{2n}$ & $\cdots$ & $-\frac{1}{2n}$ & $0$ & $\cdots$ & $0$ 
	    \end{tabular}
	    \caption{The Lower Bound Instance in Lemma \ref{lem:chores:MMS:mid$Utilitarian$:n>2:lower}}
	    \label{tab:chore-utilitarian-mms-n>3}
	\end{table}
	We can verify that $\MMS_i = -\frac{1}{n}$ for all $ i \in N$. 
	In an utilitarian welfare-maximization allocation $\mathbf{O}=(O_1,\ldots, O_n)$, the first $2n$ items are assigned to agent 1, and the rest items are arbitrarily allocated to the other agents so that $v_1(O_1)=-\frac{2}{n}$, $v_i(O_i)=0$ for $i \ge 2$ and $\UW(\mathbf{O}) = -\frac{2}{n}$. 
	However, this allocation is not MMS fair to agent 1 since $v _ 1 ( O _ 1) < \MMS_1$. 
	For any MMS allocation $\mathbf{A}$, agent 1 can receive at most $n$ of the first $2n$ items, and thus, at least $n$ of the first $2n$ items are assigned to agents $ i \geq 2 $. 
	Thus, $\UW(\mathbf{A}) \leq -\frac{1}{n} - \frac{1}{2}$, and consequently, the price of MMS is upper bounded by
	$$
	\PoF \ge \frac{\frac{1}{n} +\frac{1}{2}}{\frac{2}{n}} = \frac{1}{2} + \frac{n}{4} = \Omega(n).
	$$
	Therefore, the price of MMS with respect to utilitarian welfare is $\Theta(n).$
\end{proof}
\end{proof}

\medskip
\noindent\begin{proof}[Proof of Theorem \ref{thm::MMS-chore-RW}]
    Again, we only prove for $n \ge 3$ in this section and defer the discussion of $n=2$ to Section \ref{sec:n=2}. 
    The proof is split into the following two lemmas.

\begin{lemma}
\label{lem:chores:MMS:Egalitarian:n>2:upper}
    For the problem of Chores $\mid$ Egalitarian $\mid$ MMS with $n \ge 3$, $\PoF \le \frac{n}{2}$.
\end{lemma}

\begin{lemma}
\label{lem:chores:MMS:Egalitarian:n>2:lower}
    For the problem of Chores $\mid$ Egalitarian $\mid$ MMS with any $n \ge 3$, there is an instance where no MMS allocation has egalitarian welfare strictly greater than $\frac{2}{n}$ fraction of the maximum egalitarian welfare.
\end{lemma}


\medskip
\noindent\begin{proof}[Proof of Lemma \ref{lem:chores:MMS:Egalitarian:n>2:upper}]
    If $\OPT_E \geq -\frac{1}{n}$, then a egalitarian welfare-maximizing allocation is also MMS. If $\OPT_E \leq -\frac{2}{n}$, then the statement trivially holds due to the normalized valuations. Thus we can further assume $ -\frac{2}{n} < \OPT _ E < -\frac{1}{n}$.
    
    Let $\mathbf{A}$ be the allocation returned by Algorithm~\ref{alg:chores:MMS} and $\mathbf{A}$ is MMS according to Lemma~\ref{Lemma::MMS-UW-chore}. Renumber the agents from 1 to $n$ by the order of receiving bundles in the algorithm, where agent 1 is the first to receive a bundle and agent $n$ is the last. We first consider the case of $n=3$ and show either $\EW(\mathbf{A}) \geq -\frac{1}{2}$ holds or another allocation achieves egalitarian welfare at least $-\frac{1}{2}$. If $\min_{i}\MMS_i \ge - \frac{1}{2}$, then clearly, we have $\EW(\mathbf{A}) \ge -\frac{1}{2}$ and the statement holds. Then, we can further assume $\MMS_i < -\frac{1}{2}$ for some $i$. Denote by $\mathbf{T} = (T_1, T_2,T_3)$ the $\MMS_i$-defining partition and $v_i(T_1) \le v _ i ( T _ 2) \le v _ i ( T _ 3)$. Accordingly, we have $|T_1| = 1$ that means there exists a single chore on which agent $i$ has value less than $-\frac{1}{2}$. We then split the proof into three cases based on the number of agents with $\MMS_i < -\frac{1}{2}$.
    
    \emph{Case 1}: only one agent $x$ has $\MMS_x < -\frac{1}{2}$. Denote by $e^x$ the item with $v _ x ( e^ x) < - \frac{1}{2}$ and it suffices to show $e^x \notin A _ x$. According to the proof of Lemma~\ref{Lemma::MMS-UW-chore}, we have $ v _ 3 ( A _ 3) \ge -\frac{1}{3}$ and thus $x \neq 3$. Then, when agent $x$ receives $A _ x$, there exists another agent has not received a bundle, and due to the tie break rule of Step 5 and 8, item $e^x$ can never be assigned to agent $x$. Thus, $\PoF \le \frac{3}{2}$.
    
    \emph{Case 2}: two agents $x$, $y$ has MMS value less than $-\frac{1}{2}$. Denote by $e^i$ the item with $v _ i ( e ^ i ) < -\frac{1}{2}$ for $i = x,y$. If $e^x = e^y$, consider allocation $\mathbf{B}$ in which $e^x$ is assigned to agent $z\neq x,y$; all items at the right of $e^x$ to agent $x$; all items at the left of $e^x$ to agent $y$. We can verify that allocation $\mathbf{B}$ is MMS and has welfare $\EW(\mathbf{B}) \ge -\frac{1}{2}$. If $e^x \neq e^y$, without loss of generality, we assume $e^x$ at the left of $e^y$. Consider allocation $\mathbf{C}$ in which all items at the right of $e^x$ (not including $e^x$) are assigned to agent $x$; all other items are assigned to agent $y$. We can verify that allocation $\mathbf{C}$ is MMS and has welfare $\EW(\mathbf{C}) \ge -\frac{1}{2}$. Thus, $\PoF \le \frac{3}{2}$.
    
    \emph{Case 3}: agents $x,y,z$ have MMS value less than $-\frac{1}{2}$. Denote by $e^i$ the item with $v _ i ( e ^ i ) < -\frac{1}{2}$ for $i = x,y,z$. If $e^x = e^y = e^z$, then one can verify that allocation $\mathbf{A}$ indeed achieves $\EW(\mathbf{A}) = \OPT_E$. If items $ e^i$ are not identical, without loss of generality, we assume $e^x$ the left-most item among $e^x, e^y, e^z$. Consider the allocation $\mathbf{D}$ in which all items at the right of $e^x$ (not including $e^x$) are assigned to agent $x$; all other items are assigned to agent $y$. We can verify that allocation $\mathbf{D}$ is MMS and has welfare $\EW(\mathbf{D}) \ge -\frac{1}{2}$. Thus, $\PoF \le \frac{3}{2}$.

   	We then prove the statement for $ n \geq 4$ and it is sufficient to show that there exists a connected MMS allocation with $\EW \geq - 1/2$. If $\min _ { i \in [n]} v _ i ( A _ i ) \geq -1/2$, the statement trivially holds. We can further assume $\min _ { i \in [n]} v _ i ( A _ i ) < - 1/2$ and suppose $ v _ k ( A _ k ) \leq v _ i ( A _ i )$ for $ i \in [n]$. Then, due to condition of Step 4, bundle $ A _ k $ must be assigned in Step 9 and $ | A _ k | = 1 $. Due to normalized valuations, we have $ v _ k ( A _ i ) > - 1/2$ for any $ i \neq k $. 
   	Since Step 8 chooses the agent with the largest value, we have $v _ j ( A _ k ) < -\frac{1}{2}$ for $ j \geq k $. Accordingly, we have $ v _ n ( A _ k ) < -\frac{1}{2}$ that implies $\MMS_n < -\frac{1}{2} \leq -\frac{2}{n}$. We then analyse the value of agent $ n $ and prove that at most $\lceil \frac{n}{2} \rceil$ agents receive non-empty bundles. Denote by $\mathcal{P}$ the set of agents (excluding the one whose bundle is left-connected to $ A _ k $) whose bundle is assigned in Step 6 and $\mathcal{Q}$ the set of agents (excluding agent $k$) whose bundle is right-connected to $ A _ i $ with $i \in \mathcal{P}$. By similar argument as that in the proof of Lemma~\ref{Lemma::MMS-UW-chore}, we have
	$
	v _ n ( \cup _{i \in \mathcal{P} \cup \mathcal{Q}} A _ i ) < - \frac{|\mathcal{P} \cup \mathcal{Q}|}{n}.
	$ and $ v _ n (  A _ i ) < - \frac{2}{n}$ for any $ i \notin \mathcal{P}$. Let $\bar{\mathcal{P}} = \{ i\in [n-1] \mid A _ i \neq \emptyset \text{ and }  i \neq k\}$. We then have the following,
	$$
	\begin{aligned}
		v _ n ( \cup _ { i =1} ^ {n-1} A _ i ) &=  v _ n ( A _ k ) + v _ n ( \cup _{ i \in \mathcal{P} \cup \mathcal{Q}} A _ i ) + v _ n ( \cup _{ i \in \bar{\mathcal{P}} \setminus (\mathcal{P} \cup \mathcal{Q})} A _ i ) \\
		& < - \frac{1}{2} - \frac{| \mathcal{P} \cup \mathcal{Q}|}{n} - \frac{2(|\bar{\mathcal{P}}| - | \mathcal{P} \cup \mathcal{Q}|)}{n} \\
		& \leq -\frac{1}{2} - \frac{| \bar{\mathcal{P}}|}{n},
	\end{aligned}
	$$ where the last inequality is due to $| \bar{\mathcal{P}}| \geq | \mathcal{P} \cup \mathcal{Q}|$. As $ v_ n ( E ) = - 1 $, we have $ | \bar{\mathcal{P}}| < \frac{n}{2}$ and hence, at most $\lceil \frac{n}{2} \rceil$ agents receive non-empty bundles in allocation $\mathbf{A}$. Since $\OPT_E > -\frac{2}{n} \geq -\frac{1}{2}$, there must exist an agent $ i ^ {\prime} \in [ k -1 ]$ such that $ v _ { i ^ {\prime}} ( A _ k ) \geq - \frac{1}{2}$. We now partition $N = N _ 1 \cup N _ 2 $ where $ N _ 1 = \{ i \in [n] \mid A _ i \neq \emptyset\}$ and $ N _ 2 = N \setminus N _ 1 $. Consider a partial allocation $\mathbf{B}$ with $ B _ i = A _ i $ for $ i \in N _ 1 \setminus\{ k, i^{\prime}\}$ and $ B _ { i ^ {\prime}} = A _ k $. Then, each agent $ i \in N _ 1 \setminus\{ k \}$ has value at least $ - \frac{1}{2}$ and satisfies MMS and connectivity constraint in $\mathbf{B}$. Then note that the only unassigned bundle is $ A _{ i ^ {\prime}}$ clearly connected. Recall that $| N _ 2 \cup \{ k \}| \geq n - \lceil \frac{n}{2} \rceil  +1 $ and for each agent $ i \in N _ 2 \cup \{ k \}$, $ v _ i ( A _{ i ^ {\prime}}) > -1/2$ due to $ v _ i ( A _ k ) < - 1/2$ and $ v _ i ( E ) = - 1 $. Thus, by implementing Algorithm~\ref{alg:chores:MMS} on bundle $ A _{ i ^ {\prime}}$ and agents $ N _ 2 \cup \{ k \}$, one can extend allocation $\mathbf{B}$ to a complete MMS allocation, in which every agent has value at least $ - \frac{1}{2}$.
\end{proof}

\begin{proof}[Proof of Lemma \ref{lem:chores:MMS:Egalitarian:n>2:lower}]
	As for the lower bound, let us consider an instance with $n \ge 3$ agents and a set $E = \{ e _ 1, \ldots, e _ { n + 2 }\}$ of $ n + 2 $ items.
	The valuations are shown in the following table, where $\epsilon > 0 $ is arbitrarily small.
	\begin{table}[h]
	    \centering
	    \begin{tabular}{c|c|c|c|c|c|c|c|c}
	         Items & $e_1$ & $e_2$ & $e_3$ & $e_4$ & $e_5$ & $e_6$ & $\cdots$ & $e_{n+2}$  \\
	         \hline
	         $v_1(\cdot)$ & $-\frac{1}{n}$ & $-\epsilon$ & $-\epsilon$ & $-\frac{1}{n} + \epsilon$ & $-\frac{1}{n} + \epsilon$ & $-\frac{1}{n}$ & $\cdots$ & $-\frac{1}{n}$  \\
	         $v_i(\cdot)$ for $i\ge 2$ & $-\frac{1}{2}$ & $0$ & $-\frac{1}{2}$ & $0$ & 0 & 0 & $\cdots$ & $0$ 
	    \end{tabular}
	    \caption{The Lower Bound Instance in Lemma \ref{lem:chores:MMS:Egalitarian:n>2:lower}}
	\end{table}
	In a connected egalitarian welfare maximizing allocation $\mathbf{O}$, agent 1 receives bundle $ O _ 1 = \{ e _ 1, e _ 2, e _ 3 \}$ and has valuation $v _ 1 ( O _ 1 ) = - \frac{1}{n} - 2\epsilon$. However, one can verify that $\textnormal{MMS}_1 = - \frac{1}{n} < v _ 1 ( O_ 1 )$, and thus, allocation $\mathbf{O}$ is not MMS. 
	Thus in any MMS allocation, items $ e _ 1$ and $ e  _3 $ cannot be assigned to agent 1 at the same time. 
	Accordingly, the egalitarian welfare of a connected MMS allocation is at most $- \frac{1}{2}$, and thus the price of MMS with respect to egalitarian welfare is at least 
	$$
	\PoF \ge \frac{\frac{1}{2}}{\frac{1}{n} + 2 \epsilon} \rightarrow \frac{n}{2}, \text{ when $\epsilon \rightarrow 0 $,}
	$$
	which finishes the proof.
\end{proof}
\end{proof}

\subsection{Price of PROP1 for Indivisible chores}

\subsubsection{The Algorithm}


Algorithm \ref{alg:chores:PROP1} is similar with Algorithm \ref{alg:chores:MMS} but the critical item selected in Step \ref{step:chores:prop1:con} is changed to be the first item violating the value threshold $-\frac{1}{n}$.

\begin{algorithm}[H]
	\caption{\hspace{-2pt}{ Chores-PROP1}} 
	\label{alg:chores:PROP1}
	\begin{algorithmic}[1]
		\REQUIRE An instance $\cI = \langle N, E, \mathcal{V} \rangle$.
		\ENSURE Allocation $\mathbf{A} = (A_1,\ldots,A_n)$.
		\STATE Initialize $N _ 0 \leftarrow N$, $ E _ 0 \leftarrow E$ and $A_i = \emptyset$ for all $i \in N$.
		\WHILE{$|N _ 0| > 1 \And E _ 0 \neq \emptyset$}
		\STATE Let $ e _ L \in E $ be the left most item.
		\IF{$\exists i \in N _ 0$ such that $ v _ i(e _ L ) \geq - \frac{1}{n}$} \label{step:chores:prop1:con}
		\STATE Let $p$ be the largest index such that there exists an agent $i$ with $v _ i ( L (p) \cap E _ 0) \geq  - \frac{1}{n}$. If there is a tie, choose the agent with highest value on bundle $L(p+1)\cap E _ 0$.
		\IF{ $v _ i (L( p + 1) \cap E _ 0) < -\frac{2}{n}$}
		\STATE $ A _ i \leftarrow L (p) \cap E _ 0$. \label{step:chores:prop1:p}
		\ELSE 
		\STATE $A _ i \leftarrow L ( p + 1 ) \cap E _ 0 $. \label{step:chores:prop1:p+1}
		\ENDIF
		\ELSE
		\STATE Let $ i \in \arg\max\limits_{j\in N _ 0} v _ j ( e _ L )$, and break ties arbitrarily. Assign $ A _ i \leftarrow \{ e _ L \}.$\label{step:chores:prop1:single}
		\ENDIF
		\STATE Update $N _ 0  \leftarrow N _ 0 \setminus \{ i \}$ and $E _ 0 \leftarrow E _ 0  \setminus A_i$;
		\ENDWHILE
		\IF{$E _ 0  \neq \emptyset$}
		\STATE Let $l$ be the remaining agent in $N_0$.
		\STATE $A _ {l} \leftarrow E _ 0 $;
		\ENDIF
		\RETURN $\mathbf{A}$
	\end{algorithmic}
\end{algorithm}

\begin{lemma}\label{Lemma::PROP-chore-utilitarian}
Algorithm \ref{alg:chores:PROP1} can efficiently compute a PROP1 allocation $\mathbf{A}$.
\end{lemma}
\begin{proof}
	Notice that in every round of while-loop of Algorithm~\ref{alg:chores:PROP1}, the number of agents is reduced by one and the set of items are reduced, so the algorithm clearly terminates. Renumber the agents from $1$ to $n$ by the order of receiving bundles in the algorithm, where agent 1 is the first to receive a bundle and agent $n$ is the last.
	
	We first show that if agent $i$ who receives a bundle in while-loop, she satisfies PROP1. For the case where $A _ i $ is assigned in Step 7 or 9, denote by $ p _ i$ the index found in Step 5 at the round when $i$ receives $A_i$ and $ E ^{(i)} _ 0 $ the set of remaining items at the start of this round. Due to the condition of Step 5, we have $ v _ i ( L ( p _ i ) \cap E _ 0 ^ {(i)}) \geq - \frac{1}{n}$. If $A _ i $ is assigned in Step 7, then $ A _ i = L ( p _ i ) \cap E _ 0 ^ {(i)}$ and hence $ v _ i ( A _ i) \geq -\frac{1}{n}$. If $ A _ i $ is assigned in Step 9, we have $ A _ i = L ( p _ i + 1) \cap E _ 0 ^{(i)}$. Notice that $ v _ i (A _ i \setminus \{ e _ {p_i + 1} \}) \geq -\frac{1}{n}$ and $A _ i \setminus \{ e _ {p_i + 1} \} \in \mathcal{S}$, so we can claim that agent $i$ also satisfies PROP1. For the remaining case where $A _ i $ is assigned in Step 12, we have $| A _ i |  = 1 $ and so agent $i$ trivially satisfies PROP1. 
	Consequently, if while-loop terminates as all chores are assigned, then connected allocation $\mathbf{A}$ is PROP1. Hence, the remaining work is to show that when while-loop terminates as $ E _ 0 \neq \emptyset$ and $ |N_0| = 1$, the returned allocation $\mathbf{A}$ is still PROP1. It is sufficient to show that agent $n$ also satisfies PROP1.
	
	Denote by $\mathcal{P}$ the set of agents whose bundle is assigned in Step 7 and $\mathcal{Q}$ the set of agents whose bundle is right-connected to some $A _ i $ with $ i \in \mathcal{P}$. Note that $ |\mathcal{P} | = | \mathcal{Q}|$ and $\mathcal{P} \cap \mathcal{Q}$ may be non-empty. For an agent $ i \notin \mathcal{P}$, if bundle $A _ i $ is assigned in Step 9, we have $ v _ n ( A _ i ) < -\frac{1}{n}$; otherwise, contradicting to the choice of Step 5. If bundle $ A _ i $ is assigned in Step 12, then it also holds that $ v _ n ( A _ i ) < -\frac{1}{n}$ due to the condition of Step 4. For each $ i \in \mathcal{P}$, it uniquely maps a $ j _ i \in \mathcal{Q}$ where $ A_{ j _ i }$ is right-connected to $A_ i $. As the value of agent $n$, we have the following claim.
	
	\begin{claim}\label{claim::chores-prop1}
	Either $ v _ n ( A _ i \cup A _ { j _i}) < -\frac{2}{n}$ or $v _ n ( A _ i ) < - \frac{1}{n}$.
	\end{claim}
	
\noindent\begin{proof}[Proof of Claim~\ref{claim::chores-prop1}]
    Because if $ v _ n ( A _ i ) \geq - \frac{1}{n}$, we must have $v _ n ( A _ i \cup \{ e ^* \}) \leq v _ i ( A_ i \cup \{ e ^*\}) < -\frac{2}{n}$ where $e^ * \in A _ { j _ i }$ is the left-most item. Then, due to the monotonicity of $ v _ n (\cdot)$, we have $v _ n ( A _ i \cup A _ { j _ i }) < - \frac{2}{n}$. 
\end{proof}

Next, we bound the value of agent $n$ on the bundles received by agents $[n-1]$ and show $ v _ n ( A_ n ) > -\frac{1}{n}$. We further let $\mathcal{P} _ 0 \subsetneq \mathcal{P}$ be the set of index $ i \in \mathcal{P}$ such that $ v _ n ( A _ i ) \geq -\frac{1}{n}$ but $ v _ n ( A _ i \cup A _{ j _ i }) < -\frac{2}{n}$ where $ A _ { j _ i }$ with $ j _ i \in \mathcal{Q}$ is right-connected to $A _ i $. Denote by $ \mathcal{Q} _ 0 \subsetneq \mathcal{Q}$ be the set of agents-index whose bundle is righted connected to bundle $ A _ i $ for some $ i \in \mathcal{P} _ 0 $. Similarly, $ | \mathcal{ P } _ 0 | = | \mathcal{ Q } _ 0 | $ and $ \mathcal{ P } _ 0 \cap \mathcal{ Q } _ 0 $ can be non-empty. Note that $ v _ n ( A _ i ) < -\frac{1}{n}$ holds for $ i \notin \mathcal{P} _ 0 $. 

	Note that $\cup_{i \in \mathcal{P}_0 \cup \mathcal{Q} _ 0 } A _ i $ contains several maximal connected blocks $\mathcal{A}_s$ on $E$, and each maximal block $\mathcal{A}_s$ is the union of at least two bundles in $\{ A_i \}_{i \in \mathcal{P}_0 \cup \mathcal{Q} _ 0}$. In particular, $\mathcal{A}_s$ 
	is in the form: the left most bundle $A _ p $ satisfies $ p \in \mathcal{P}_ 0 \setminus \mathcal{Q} _ 0 $; the right-most bundle $A _ q $ satisfies $ q \in \mathcal{Q} _0\setminus \mathcal{P} _ 0 $; other bundles $ A_  r $ satisfies $r \in \mathcal{P} _ 0 \cap \mathcal{Q} _ 0 $. Now, suppose that $ \mathcal{A}_ s$ contains a number $ k_s $ of bundles $ \{ A _ i \} _ { i \in \mathcal{P} _ 0 \cup \mathcal{Q} _ 0 }$. Recall that for $ i \in \mathcal{P} _0 $ and $ j _ i \in \mathcal{Q} _ 0 $ with $ A _ {j _ i }$ right-connected to $A _ i $, we have $ v _n ( A_  i \cup A _{j_i }) < -\frac{2}{n}$. Then, if $ k _ s$ is even, we have
	$$
	v _ n ( \mathcal{A}_s) < -\frac{k _ s}{2} \times \frac{2}{n} = - \frac{k _ s}{n}.
	$$
	And if $k _ s$ is odd, since the right most bundle $A_q$ does not in $\{ A_ i \} _ { i \in \mathcal{P}_0 }$, i.e., $v _ n (A_ q) < - \frac{1}{n}$, then,
	$$
	v _ n ( \mathcal{A}_s) < - \frac{k _ s -1}{2} \times \frac{2}{n} - \frac{1}{n} = - \frac{k_s}{n}.
	$$
	Therefore, we have $ v _ n (\mathcal{A} _ s ) < - \frac{|\mathcal{A} _s |}{n}$ for all $\mathcal{A}_s$. Summing over all the maximal connected blocks $\mathcal{A}_s$, we have
	$$
	v _ n ( \cup _ { i \in \mathcal{P}_0 \cup \mathcal{Q} _ 0 } A _ i ) = \sum _ { s }v _ n ( \mathcal{A} _ s) < -\sum _{s } \frac{| \mathcal{A} _ s |}{n} = - \frac{ | \mathcal{P}_0 \cup \mathcal{Q}_0|}{n}.
	$$
    Due to $v _ n ( A _ i ) < -\frac{1}{n}$ for all $i\notin \mathcal{P}_0$. we have,
	$$
	\begin{aligned}
		v _ n ( \cup_{i \in [n-1]} A _ i ) &=  \sum _ { i \in  \mathcal{P} _ 0  \cup \mathcal{Q} _ 0 } v _ n ( A_ i ) + \sum _ { i \in [n-1]\setminus (\mathcal{P} _ 0 \cup \mathcal{Q} _ 0 )}   v _ n ( A _ i ) \\
		& < -\frac{| \mathcal{P} _ 0 \cup \mathcal{Q} _ 0 |}{n} - \frac{n-1 - | \mathcal{P} _ 0 \cup \mathcal{Q} _ 0 |}{n} \\
		& = - \frac{n-1}{n}.
	\end{aligned}
	$$
	Then, due to the normalized valuations, we have $v _ n ( A _ n ) \geq - \frac{1}{n}$, which completes the proof of Lemma~\ref{Lemma::PROP-chore-utilitarian}.
\end{proof}

\subsubsection{Chores $\mid$ Utilitarian $\mid$ PROP1}

\begin{proof}[Proof of Theorem \ref{thm::PROP1-chore-UW-n>3}] Again, we only prove for $n\ge 3$ in this section and defer the discussion of $n=2$ to Section~\ref{sec:n=2}. The proof is split into the following two lemmas.

\begin{lemma}\label{lemma::PROP1-chore-utilitarian-ak-add-1}
    For the problem of Chores $\mid$ Utilitarian $\mid$ PROP1 with $n \ge 3$, $\PoF \leq 3n$.
\end{lemma}

\begin{lemma}\label{lemma::PROP1-chore-utilitarian-ak-add-2}
    For the problem of Chores $\mid$ Utilitarian $\mid$ PROP1 with $n \ge 3$, $\PoF \ge \frac{n}{6}$.
\end{lemma}

\noindent\begin{proof}[Proof of Lemma~\ref{lemma::PROP1-chore-utilitarian-ak-add-1}]
   We consider the allocation $\mathbf{A}$ returned by Algorithm \ref{alg:chores:PROP1}, and according to Lemma~\ref{Lemma::PROP-chore-utilitarian}, allocation $\mathbf{A}$ is PROP1. We now let $ {N} _ 1, N_2, N_3$ be the set of agents who receive items in Step 7, 9, and 12, respectively. Then $ v _ i ( A_ i ) \geq -\frac{1}{n}$ holds for all $ i \in N_ 1$ and $ v _ i ( A_ i ) \geq -\frac{2}{n}$ holds for all $ i \in N_2$. We now bound the sum of value of agents in $N _ 3$. Suppose $N _ 3 = \{ i _ 1, i _2, \ldots, i _ p \}$ and bundle $ A _{i_ l}$ is on the left of bundle $ A _ { i _ k }$ for any $ l < k \leq p $. For every $ k \leq p $, we have $ v _ { i _ p }( A _ k ) \leq v _ { i _ k } ( A _ { i _ k })$ due to the condition in Step 12 of Algorithm \ref{alg:chores:PROP1}. Consequently, the welfare of agents in $N_3$ is bounded by
	$$
	 \sum _{ i \in N _3 }v _ i ( A _ i ) \geq \sum_{ i\in N _3} v _{ i _ p } ( A _ i ) \geq -1,
	$$
	where the last transition is due to the additivity of $ v _ { i _ p } (\cdot)$ and $v _ {i _ p } ( E ) = -1$. As for agent $n$, according to the proof of Lemma~\ref{Lemma::PROP-chore-utilitarian}, we have $v _ n ( A _ n ) \geq - \frac{1}{n}$. Therefore, the utilitarian welfare of allocation $\mathbf{A}$ satisfies the following,
	$$
	\begin{aligned}
		\UW(\mathbf{A}) & \geq ( \sum_{ i \in N_1} + \sum_{i \in N_2} + \sum_{ i \in N_3} ) v _ i ( A_  i)  + v _ n ( A _ n ) \\
		& > -\frac{2}{n}(|N _ 1| + |N_ 2 | ) - 1 - \frac{1}{n}\\
		& \geq -3,
	\end{aligned}
	$$ where the last inequality is due to when either term -1 or $ - \frac{1}{n}$ exits, it holds that $ |N_1| + |N _ 2 | \leq n - 1$. Suppose $\mathbf{O}$ be a connected utilitarian welfare-maximizing allocation. If $\UW(\mathbf{O}) \geq -\frac{1}{n}$, then allocation $\mathbf{O} $ is PROP1 and the statement trivially holds. We can further assume $\UW(\mathbf{O} ^ {\prime}) < -\frac{1}{n}$. Since allocation $\mathbf{A}$ is a connected PROP1 allocation with welfare at least $-3$, the price of PROP1 is at most $3n$.
\end{proof}
	
\medskip
\noindent\begin{proof}[Proof of Lemma~\ref{lemma::PROP1-chore-utilitarian-ak-add-2}]
  \begin{table}[h]
	    \centering
	    \begin{tabular}{c|c|c|c|c|c}
	         Items & $e_1$ & $e_2$ & $\cdots$ & $e_n$ & $e_{n+1}$ \\
	         \hline
	         $v_1(\cdot)$ & $-\frac{2}{n^2}$ & $-\frac{2}{n^2} $ & $\cdots$ & $-\frac{2}{n^2} $ & $-\frac{n-2}{n}$ \\
	         $v_i(\cdot)$ for $i\ge 2$ & $- \frac{1}{n+1}$ & $- \frac{1}{n+1}$ & $\cdots$ & $- \frac{1}{n+1}$ & $- \frac{1}{n+1}$
	    \end{tabular}
	    \caption{The Lower Bound Instance for PROP1 in Lemma \ref{lemma::PROP1-chore-utilitarian-ak-add-2}}
	    \label{tab::chore-utilitarian-prop1-n>3}
	\end{table}
    	As for the lower bound, consider an instance with $n$ (even) agents and a set $E = \{ e _ 1, \ldots, e _ { n + 1} \}$ of $ n + 1 $ chores. The valuations are shown in Table \ref{tab::chore-utilitarian-prop1-n>3}. In an utilitarian welfare-maximizing allocation $\mathbf{O} = ( O _ 1, \ldots, O _ n )$, the first $n$ items are assigned to agent 1 and item $e_{n+1}$ is arbitrarily allocated to the other agents so that $\UW(\mathbf{O}) = -\frac{3n+2}{n(n+1)}$. However, agent 1 violates PROP1 in $\mathbf{O}$. In a PROP1 allocation $\mathbf{A}$, agent 1 can receive at most $\frac{n}{2} + 1$ items, which leaves at least $n$ items to other agents. Thus, $\UW(\mathbf{A}) \le -\frac{n}{2(n+1)} - \frac{n+2}{ n ^ 2}$, and consequently, the price of PROP1 is upper bounded by
    	$$
    	\PoF \geq \frac{n^3 + 2(n+2)(n+1)}{2n(3n+2)} \ge \frac{n}{6} + \frac{2}{9} = \Omega(n).
    	$$
\end{proof}

Combining Lemma~\ref{lemma::PROP1-chore-utilitarian-ak-add-1} and Lemma~\ref{lemma::PROP1-chore-utilitarian-ak-add-2}, the price of PROP1 with respect to utilitarian welfare is $\Theta(n)$.
\end{proof}

\subsubsection{Chores $\mid$ Egalitarian $\mid$ PROP1}

\begin{proof}[Proof of Theorem \ref{thm::PROP1-chore-RW-n>3}]
We start from the upper bound. If $\OPT _ E \geq -\frac{1}{n}$, then an egalitarian welfare maximizing allocation is also PROP1 and we are done. 
Moreover, since the valuations are normalized, 
any allocation has egalitarian welfare no smaller than $-1$, and thus if $ \OPT _ E \leq -\frac{2}{n}$, the statement trivially holds. 
We then focus on the case when $-\frac{2}{n} < \OPT _ E < - \frac{1}{n}$. Denote by $\mathbf{A}$ the allocation returned by Algorithm~\ref{alg:chores:PROP1}, and without loss of generality, we assume $ 1, \ldots, n$ is the order of agents receiving items. If several agents receive none, the order among them is arbitrary. According to Lemma~\ref{Lemma::PROP-chore-utilitarian}, allocation $\mathbf{A}$ is PROP1.

We first prove the statement for the case $ n = 3 $, and it is sufficient to show the existence of a connected PROP1 allocation with egalitarian welfare at least $ -\frac{2}{3}$. Based on the proof of Lemma~\ref{Lemma::PROP-chore-utilitarian}, we have $ v _ 3 ( A _ 3 ) \geq -\frac{1}{3}$. We then discuss the assignment of the first two agents. If neither of $ A _ 1$ nor $ A _ 2 $ is assigned in Step 12, we clearly have $ v _ i ( A _ i ) \geq - \frac{2}{3}$ for any $ i \in [2]$, then $\EW(\mathbf{A}) \geq - \frac{2}{3}$ holds. If $ A _ 1 $ is assigned in Step 12, we must have $ - \frac{2}{3} < \OPT_E \leq v _ 1 ( A _ 1 ) < - \frac{1}{3}$ due to the condition of Step 4 and 12. Moreover, under this case, we must have $v _ 2 ( A_ 2) \geq - \frac{2}{3}$ due to the normalized valuation. Thus, $\EW(\mathbf{A}) \geq - \frac{2}{3}$ holds. If only $ A _ 2 $ is assigned in Step 12, we can only consider the case $ v _ 2 ( A _ 2) < - \frac{2}{3}$. Recall that under this case $|A _ 2 | = 1$ holds, and let $ i ^*$ be the agent who receives bundle $ A _ 2 $ in an egalitarian welfare-maximizing allocation. Then, swapping the bundle received by agent 2 and $ i ^*$ leads to a connected PROP1 allocation with egalitarian welfare at least $-\frac{2}{3}$. For the lower bound, let us consider an instance with three agents and a set $E = \{ e _ 1, \ldots, e _ 7 \}$ of seven items. The valuations are shown in the following table, in which $\epsilon >0$ is arbitrarily small.
\begin{table}[h]
	    \centering
	    \begin{tabular}{c|c|c|c|c|c|c|c}
	         Items & $e_1$ & $e_2$ & $e_3$ & $ e _4$ & $ e_5$ & $e _ 6 $ & $e_7$   \\
	         \hline
	         $v _ 1$ & $ 0 $ & $ 0 $ & $ 0 $ & $- \frac{1}{3}$ &  $- \epsilon$ & $ -\frac{1}{3} - 5\epsilon$ & 
	         $ -\frac{1}{3} + 6 \epsilon$  \\
	         $v_i(\cdot) \textnormal{ for } i = 2,3$ & $-\epsilon$ & $ -\frac{1}{3}$ & $-\epsilon$ & $ -\epsilon $ & $ -\epsilon$ & $ -\frac{2}{3} + 4\epsilon $ & $0$ 
	    \end{tabular}
	    \caption{The Lower Bound Instance for $n =3 $ in Theorem \ref{thm::PROP1-chore-RW-n>3}}
	\end{table}
	
Consider an egalitarian welfare-maximizing allocation $\mathbf{O}$ with $O_ 1 = \{ e_6\}$, $ O _ 2 = \{ e_1, \ldots, e _ 5 \}$ and $ O _ 3 = \{ e _ 7 \}$. We can compute $\EW(\mathbf{O}) = -\frac{1}{3} - 5\epsilon$ and thus, $\OPT _ E \geq -\frac{1}{3} - 5\epsilon$ holds. Then, for an PROP1 allocation $\mathbf{D}$, one can verify that either $e_6 \in D_2\cup D _3$ or $ e _ 6 \cup e _ 7 \in D_1$ holds. In both cases, we have $\EW(\mathbf{D}) \geq -\frac{2}{3 } + \epsilon$. Therefore, the price of PROP1 with respect to egalitarian welfare is bounded by$$
\PoF \geq \frac{-\frac{2}{3} + \epsilon}{ -\frac{1}{3} - 5\epsilon} \rightarrow 2, \textnormal{ when } \epsilon \rightarrow 0.
$$

Next we prove the statement for $ n \geq 4$. 
Since $\OPT_E < -\frac{1}{n}$, if Algorithm \ref{alg:chores:PROP1} returns a PROP1 allocation $\cA$ with egalitarian welfare at least $ - \frac{1}{2}$, then we finish the proof. 
If $\min_{ i \in [n]} v _ i ( A _ i ) < - \frac{1}{2}$, we next construct another allocation which satisfies the requirement. 
Let agent $b$ be the one that has smallest value, i.e., $ v _ k (  A _ k ) \leq v _ i ( A _ i )$ for $ i \in N$. 
By the design of the algorithm, $ A _ k $ must be assigned in Step \ref{step:chores:prop1:single} and $ | A _ k | = 1 $. Due to the normalized allocation, we can claim that $ v _ k ( A _ i ) > - \frac{1}{2}$ for $ i \neq k$. 
According to the tie break rule in Step \ref{step:chores:prop1:single}, we have $ v _ j ( A _ k )< -\frac{1}{2}$ for $ j > k $; otherwise Step \ref{step:chores:prop1:single} would pick other agents, rather than agent $k$. 
We now analyze the value of agent $ n $ and prove that at most $\lceil \frac{n}{2} \rceil$ agents receive non-empty bundles in $\mathbf{A}$. 
Let $\mathcal{P}_0$ be the set of agents (excluding the one whose bundle is on the left of $A _ k$ and connected to $ A _ k $ if any) receiving bundles in Step \ref{step:chores:prop1:p} and thus $ v _ n ( A _ i ) \ge -\frac{1}{n} $ for all $ i \in \mathcal{P} _ 0$. 
Let $\mathcal{Q} _ 0 $ be the set of agents whose bundles are on the right of $A _ i$ and connected to $A _ i $ for some $ i \in \mathcal{P} _ 0 $. 
Note that $|\mathcal{P} _ 0 | = | \mathcal{Q} _ 0 |$ and $\mathcal{P} _ 0 \cap \mathcal{Q} _ 0 $ can be non-empty. 
One can verify that for $ i \in \mathcal{P} _ 0$, we have $ v _ n ( A _ i\cup A _ { j _ i }) < -\frac{2}{n}$ where $A _ {j _ i }$ is on the right of $A _ i$ and $ A _ i\cup A _ { j _ i } \in \mathcal{S}$. 
By similar argument with that in the proof of Lemma~\ref{Lemma::PROP-chore-utilitarian}, it holds that 
\[
v _n ( \cup _{ i \in \mathcal{P} _ 0 \cup \mathcal{Q} _ 0 } A _ i ) < -\frac{|\mathcal{P} _ 0 \cup \mathcal{Q} _ 0|}{n}.`
\]
To bound the value of agent $n$, we define ${\mathcal{R} _ 0 } = \{ i\in [n-1] \mid A _ i \neq \emptyset \text{ and }  i \neq k\}$. 
Notice that $ v _ n ( A _ i ) < -\frac{1}{n}$ holds for all $i \notin \mathcal{P} _ 0 $, and thus we have the following
$$
\begin{aligned}
	 v _ n ( \bigcup_ { i \in [ n - 1]} A _ i ) &=  v _ n ( A _ k ) + v _ n ( \bigcup _{ i \in \mathcal{P} _ 0 \bigcup \mathcal{Q} _ 0 } A _ i ) + v _ n ( \bigcup _{ i \in {\mathcal{R} _ 0} \setminus (\mathcal{P}_0 \cup \mathcal{Q}_0)} A _ i ) \\
	 & < - \frac{1}{2} - \frac{| \mathcal{P} _ 0 \cup \mathcal{Q} _ 0|}{n} - \frac{|{\mathcal{R}_0}| - | \mathcal{P}_0 \cup \mathcal{Q}_0|}{n} \\
	 & = -\frac{1}{2} - \frac{| {\mathcal{R}_0}|}{n}.
\end{aligned}
$$
Since $v_i(\bigcup_{i\in [n]} A_i)\ge -1$, $ | {\mathcal{R} _ 0 }| < \frac{n}{2}$ and hence, at most $\lceil \frac{n}{2} \rceil$ agents receive non-empty bundles in allocation $\mathbf{A}$. 
Since $\OPT _ E > -\frac{2}{n} \geq -\frac{1}{2}$, there must exist an agent $ i ^ {\prime} \in [ k -1 ]$ such that $ v _ { i ^ {\prime}} ( A _ k ) \geq - \frac{1}{2}$. We now partition $N = N _ 1 \cup N _ 2 $ where $ N _ 1 = \{ i \in [n] \mid A _ i \neq \emptyset\}$ and $ N _ 2 = N \setminus N _ 1 $. Consider a partial allocation $\mathbf{B}$ with $ B _ i = A _ i $ for $ i \in N _ 1 \setminus\{ k, i^{\prime}\}$ and $ B _ { i ^ {\prime}} = A _ k $. Then, each agent $ i \in N _ 1 \setminus\{ k \}$ has value at least $ -\frac{1}{2}$ and satisfies PROP1 and connectivity constraint in $\mathbf{B}$. Then, the only unassigned bundle is $ A _{ i ^ {\prime}}$ that is clearly connected. Recall that $| N _ 2 \cup \{ k \}| \geq n - \lceil \frac{n}{2} \rceil  +1 $ and for each agent $ i \in N _ 2 \cup \{ k \}$, $ v _ i ( A _{ i ^ {\prime}}) > -\frac{1}{2}$ due to $ v _ i ( A _ k ) < - \frac{1}{2}$ and $ v _ i ( E ) = - 1 $. 
Thus, by implementing Algorithm~\ref{alg:chores:PROP1} on bundle $ A _{ i ^ {\prime}}$ and agents $ N _ 2 \cup \{ k \}$, one can extend allocation $\mathbf{B}$ to a complete PROP1 allocation, in which, moreover, every agent has value at least $ - \frac{1}{2}$.

For the lower bound, consider the instance with $n$ agents and a set $ E = \{ e_1, \ldots, e _ { n + 2} \} $ of $ n + 2 $ chores. The valuations are shown in the following table, where $\epsilon > 0 $ is arbitrarily small. 
 \begin{table}[h]
	    \centering
	    \begin{tabular}{c|c|c|c|c|c|c|c|c}
	         Items & $e_1$ & $e_2$ & $e_3$ & $ e _4$ & $ \cdots$ & $e _ n $ & $e_{n+1}$ & $ e _ { n + 2 }$  \\
	         \hline
	         $v_1(\cdot)$ & $-\epsilon$ & $- \frac{1}{n}$ & -$\epsilon$ &  $-\frac{1}{n}$ & $ \cdots$ & 
	         $ -\frac{1}{n}$ & $-\frac{1}{n}$ & $-\frac{1}{n} + 2\epsilon$\\
	         $v_i(\cdot) \textnormal{ for } i\ge 2$ & $-\frac{1}{2}$ & $0$ & $-\frac{1}{2}$ & $ 0 $ & $ \cdots$ & $ 0 $ & $0$ & $0$
	    \end{tabular}
	    \caption{The Lower Bound Instance for $n \geq 4 $ in Theorem \ref{thm::PROP1-chore-RW-n>3}}
	\end{table}
	In a egalitarian welfare-maximizing allocation $\mathbf{O}$, agent 1 receives bundle $ O _ 1 = \{ e_1, e_2, e_3 \} $ and has value $ - \frac{1}{n} - 2\epsilon$, and agent $ i \geq 2$ receive other items and have value zero. Then, we have $\OPT_ E = -\frac{1}{n} - 2\epsilon$. However, in allocation $\mathbf{O}$, agent 1 violates PROP1 as removing item $e_1$ or $e_3$ still yields value $-\frac{1}{n} - \epsilon < - \frac{1}{n} $ for him, and thus, in a PROP1 allocation, agent 1 can not receive all $ e_1, e_2, e_3$. In a PROP1 allocation $\mathbf{A}$, at least one of $ e_1, e_3$ must be assigned to agents $ i\ge 2$, which implies $\EW(\mathbf{A}) \leq -\frac{1}{2}$. Therefore, the price of PROP1 with respect to egalitarian welfare is at least $$\PoF \geq \frac{-\frac{1}{2}}{ -\frac{1}{n} - 2\epsilon} \rightarrow  \frac{n}{2}, \textnormal{ when } \epsilon \rightarrow 0 ,$$ which completes the proof.
\end{proof}

\section{Missing Proofs in Section \ref{sec:n=2}}

\begin{proof}[Proof of Theorem \ref{thm:n=2:goods}] 

    Consider the allocation $\mathbf{O}$ constructed in Lemma \ref{lem:n=2:goods}.
    The PoF regarding egalitarian welfare is straightforward by the design, and we focus on the  utilitarian welfare in the following.
    By Lemma \ref{lem:n=2:goods}, we have $\UW(\mathbf{O}) \ge 1$.
    Let $\OPT_U$ be the optimal utilitarian welfare, and thus $\OPT_U \ge 1$.
    Denote by $\OPT_U = 1 + x$ where $x\ge 0$.
    If $x \le \frac{1}{2}$, since $\UW(\mathbf{O}) \ge 1$,  we have the desired PoF ratio of $\frac{3}{2}$;
    If $x > \frac{1}{2}$, then the optimal egalitarian welfare is at least $x$ and
    $\UW(\mathbf{O}) \ge 2x$.
    Thus
    \[
    \PoF \le \frac{1+x}{2x} = \frac{1}{2x} + \frac{1}{2} < \frac{3}{2},
    \]
    which completes the proof of the upper bounds.
    
    Next we prove the tightness. 
    Regarding MMS, let us consider an instance the following instance with three items $ E = \{ e_1, e_2, e_3 \}$ of three items. 
    The valuations are shown in Table \ref{tab::thm24-mms} in which $\epsilon>0$ is arbitrarily small.
    It can be verified that a utilitarian welfare maximizing allocation $\mathbf{A}$ assigns $ \{e _ 1, e _ 2\}$ to agent 1 and $ \{e _ 3\}$ to agent 2 resulting in the welfare of $ \frac{3}{2} - 2\epsilon$. 
    However, $ v _ 2 ( A _ 2 ) < \frac{1}{2} = \MMS_2$ and thus $\mathbf{A}$ is not MMS. 
    Actually, in an MMS allocation $\mathbf{B}$, it must be the case that one agent receives $ e _ 1 $ and the other agent receives the rest two items, which implies $\UW(\mathbf{B}) = 1$. 
    Therefore, the price of MMS is at least $ \frac{\frac{3}{2} - 2\epsilon}{1} \rightarrow \frac{3}{2}$ as $\epsilon \rightarrow 0 $.  
    
    \begin{table}[h]
	    \centering
	    \begin{tabular}{c|c|c|c}
	         Items & $e_1$ & $e_2$ & $e_3$ \\
	         \hline
	         $v_1(\cdot)$ & $\frac{1}{2}$ & $\frac{1}{2}-\epsilon$ & $\epsilon$  \\
	         $v_2(\cdot)$ & $\frac{1}{2}$ & $\epsilon$ & $\frac{1}{2}-\epsilon$ 
	    \end{tabular}
	    \caption{The Lower Bound Instance for MMS fairness in Theorem \ref{thm:n=2:goods}}
	    \label{tab::thm24-mms}
	\end{table}

    \begin{table}[h]
	    \centering
	    \begin{tabular}{c|c|c|c|c}
	         Items & $e_1$ & $e_2$ & $e_3$ & $e_4$  \\
	         \hline
	         $v_1(\cdot)$ & $\epsilon$ & $\epsilon$ & $\frac{1}{2}-\epsilon$ & $\frac{1}{2}-\epsilon$ \\
	         $v_2(\cdot)$ & $\frac{1}{4}-\epsilon$ & $\frac{1}{4}-\epsilon$ & $\epsilon$ & $\frac{1}{2}+\epsilon$
	    \end{tabular}
	    \caption{The Lower Bound Instance for PROP1 fairness in Theorem \ref{thm:n=2:goods}}
	    \label{tab::thm24-prop1}
	\end{table}
    
    Regarding PROP1, let us consider an instance with 4 items $E = \{ e_1, e_2, e_3, e_4 \}$. The valuations are shown in Table \ref{tab::thm24-prop1}. 
    It can be verified that $\bar{\mathbf{A}}$ with $ \bar{A} _ 1 = \{ e_3, e_4\}$ and $ \bar{A} _ 2 = \{ e_1, e_2 \}$ achieves the optimal utilitarian welfare $\UW(\bar{\mathbf{A}}) = \frac{3}{2} - 4\epsilon$. 
    But $\bar{\mathbf{A}}$ is not PROP1 as $v_2(A_2\cup \{ e_3\}) < \frac{1}{2}$. 
    Further, it is not hard to check that $\bar{\mathbf{B}}$ with $\bar{B}_1 = \{ e_1, e_2, e_3\}$ and $ \bar{B}_2 = \{ e_4 \} $ is a PROP1 allocation that achieves the largest welfare of $1 + 2\epsilon$. 
    Taking $\epsilon \rightarrow 0 $, the price of PROP1 is at least $\frac{3}{2}$.
\end{proof}

\medskip
\noindent\begin{proof}[Proof of Lemma \ref{lem:n=2:chores}]
  We explicitly construct such an allocation $\mathbf{O} = (O_1,\ldots, O_n)$:
    \begin{itemize}
        \item $\mathbf{O}$ first maximizes the egalitarian welfare among all connected allocations; If there is a tie, $\mathbf{O}$ minimizes the number of items allocated to the agent with smaller value.
    \end{itemize}
By construction, it is straightforward that $\mathbf{O}$ maximizes the egalitarian welfare. With loss of generality, we assume $ v _ 1 ( O _ 1 ) \leq v _ 2 ( O _ 2 )$ and $ O _ 1 $ is at the left of $ O _ 2$. If $ v _ 1 ( O _ 1) \ge -\frac{1}{2}$, allocation $\mathbf{O}$ satisfies the conditions described in the statement. If $ v _ 2 ( O _ 2) < -\frac{1}{2}$ or $ v _ 2 ( O _ 1) > v _ 1 ( O _ 1)$, swapping the bundles would increase the egalitarian welfare. In the following, it suffices to consider the case where $ v _ 2 ( O _ 1) \le v _ 1 ( O _ 1) < -\frac{1}{2} \le v _ 2 ( O _ 2)$ and show $ O _ 1 $ satisfies MMS and PROP1 for agent 1.

First, we show $\mathbf{O}$ is PROP1 for agent 1.  Let $e^* \in O _ 1$ be the item such that $O _ 2 \cup\{ e ^* \}$ is connected. Since $\mathbf{O}$ minimizes the number of items assigned to smaller value agent, we have $ v _ 1 ( e ^ *) <  0 $. Consider another connected allocation $\mathbf{O}^{\prime} = (O ^ {\prime} _ 1, O ^ {\prime} _ 2 )$ with $ O^{\prime} _ 1 = O _ 1\setminus \{ e ^* \}$ and $ O^{\prime} _ 2 = O _ 2 \cup \{e^* \}$. Since $ v _ 1 (O^{\prime} _ 1) > v _ 1 ( O _ 1 )$, we must have $ v _ 2 ( O^{\prime} _ 2 ) \leq v _ 1 ( O _ 1 ) < -\frac{1}{2}$; otherwise $\EW(\mathbf{O}^{\prime}) > \EW(\mathbf{O})$, contradiction. As a consequence, $ v _ 1 ( O^{\prime} _ 1) >- \frac{1}{2}$ must hold; otherwise, swapping $O^{\prime}_1$ and $O^{\prime}_2$ results in another allocation with egalitarian welfare at least $-\frac{1}{2}$, contradiction. Thus, agent 1 satisfies PROP1 under allocation $\mathbf{O}$.

Next, we show $\mathbf{O}$ is MMS fair for agent 1. Let $\mathbf{T} = \{ T _ 1, T _ 2 \}$ be an $\MMS_1$-defining partition and $ T _ 1$ is at the left of $ T _ 2$. Again, we consider allocation $\mathbf{O}^{\prime}$ constructed in the previous paragraph. If $ v _ 1 ( O _ 1 ) < \MMS_1$, then we have $ T _ 1 \subsetneq O _ 1 $ and accordingly $ T _ 1 \subseteq O^{\prime} _ 1  $. Since $ v _ 1 ( O ^ {\prime} _ 1) > -\frac{1}{2}$, we have $ v_ 1 ( T _ 1) \ge v _ 1 ( O ^ {\prime} _ 1) > -\frac{1}{2}$ and $\MMS_1 = v _ 1 ( T _ 2) \le v _ 1 (O^{\prime} _ 2)$. Thus, $ v _ 1 ( O ^ {\prime} _ 2) > v _ 1 ( O _ 1)$. That is by allocating $ O ^ {\prime} _ 2$ to agent 1 and $ O ^ {\prime} _ 1$ to agent 2, we have $ v _ 2 ( O^ {\prime} _ 1) > -\frac{1}{2} > v _ 1 (O_1)$ and $ v _ 1 ( O^ {\prime} _ 2) > v _ 1 ( O _ 1)$, which induces higher egalitarian welfare than $\EW(\mathbf{O})$ and leads to a contradiction.

Finally, we show $\UW(\mathbf{O}) \ge -1$. Since $ v _ 2 ( O _ 1 ) \le v _ 1 ( O _ 1) $, then $ v _ 2 ( O _ 2) = - 1 - v _ 2 ( O _ 1 ) \ge - 1 - v _ 1 ( O _ 1)$ and thus $ v _1 ( O _ 1) + v _ 2 ( O _ 2) \ge 1$, which completes the proof of the lemma.
\end{proof}

\medskip
\noindent\begin{proof}[Proof of Theorem \ref{thm:n=2:chores}]
    Consider the allocation $\mathbf{O}$ constructed in Lemma \ref{lem:n=2:chores}.
    The PoF regarding egalitarian welfare is straightforward by the design, and we focus on the  utilitarian welfare in the following.
    By Lemma \ref{lem:n=2:chores}, we have $\UW(\mathbf{O}) \ge -1$.
    Let $\OPT_U$ be the optimal utilitarian welfare, and thus $\OPT_U \ge -1$.
    Denote by $\OPT_U = -1 + x$ where $x\ge 0$.
    If $x \le \frac{1}{2}$, then $\OPT_U\le -\frac{1}{2}$ and since $\UW(\mathbf{O}) \ge -1$,  we have the desired PoF ratio of $2$;
    If $x > \frac{1}{2}$, then the optimal egalitarian welfare is at least $-1+x$ and
    $\UW(\mathbf{O}) \ge 2(-1+x)$, which also gives the ratio of $2$.

	Next, we prove the lower bound.
	Regarding MMS, consider an instance with two agents and three items $ E = \{ e _ 1, e _ 2, e _ 3 \}$.
	The valuations are shown in the following table, where $\epsilon > 0 $ is arbitrarily small.
	\begin{table}[h]
	    \centering
	    \begin{tabular}{c|c|c|c}
	         Items & $e_1$ & $e_2$ & $e_3$ \\
	         \hline
	         $v_1(\cdot)$ & $-\frac{1}{2}$ & $-\frac{1}{2} + \epsilon$ & $-\epsilon$  \\
	         $v_2(\cdot)$ & $-\frac{1}{2}$ & $-\epsilon$ & $-\frac{1}{2} + \epsilon$
	    \end{tabular}
	    \caption{The Lower Bound Instance for MMS fairness in Theorem \ref{thm:n=2:chores}}
	    \label{tab:my_label}
	\end{table}
	Note that $ \MMS _ i= -\frac{1}{2}$ for both $ i \in [2]$. 
	The utilitarian welfare-maximizing allocation $\mathbf{A}$ assigns chores $\{e_1, e_2\}$ to agent 2 and $\{ e _3\} $ to agent 1 resulting the maximum welfare of $-\frac{1}{2} - 2\epsilon$. 
	However, the allocation is not MMS fair to agent 2. 
	Actually, in any MMS allocation, it must be the case that an agent receives $e _ 1 $ and the other one receives $ e _2, e_3$, and thus the welfare of an MMS allocation is $-1$. 
	Therefore, the price of MMS with respect to utilitarian welfare is bounded by
	\[
	\PoF \ge \frac{1}{1/2 + 2\epsilon} \rightarrow 2, \text{ when } \epsilon \rightarrow 0,
	\] 
	which finishes the proof.
	
	Regarding PROP1 fairness, consider an instance with 7 items $E = \{ e_1,\ldots, e_7 \}$. Agents' valuations are described in the following table,
where $\epsilon > 0 $ takes arbitrarily small value.

\begin{table}[h]
	    \centering
	    \begin{tabular}{c|c|c|c|c|c|c|c}
	         Items & $e_1$ & $e_2$ & $e_3$ & $e_4$ & $e_5$ & $e_6$ & $e_7$ \\
	         \hline
	         $v_1(\cdot)$ & $-\epsilon$ & $-\epsilon$ & $- \frac{1}{2}$ & $-\epsilon$ & $- \frac{1}{2} + 5\epsilon$ & $-\epsilon$ & $-\epsilon$ \\
	         $v_2(\cdot)$ & $-\frac{1}{4} + \epsilon$ & $-\frac{1}{4} + \epsilon$ & $0$ & $-\frac{1}{8} - \epsilon$ & $-\frac{1}{8} - \epsilon$ & $-\frac{1}{4} + \epsilon$ & $- \epsilon$
	    \end{tabular}
	    \caption{The Lower Bound Instance for PROP1 fairness in Theorem \ref{thm:n=2:chores}}
	    \label{tab:my_label}
	\end{table}

Denote by $\mathbf{A} ^*$ the allocation with optimal utilitarian welfare. Clearly, we have $ A^* _ 1 = \{ e_6, e_7\}$ and $ A^*_2 = E \setminus A^* _ 1$, and welfare $\UW(\mathbf{A}^*) = -\frac{1}{2}- 2\epsilon$. 
But in $\mathbf{A} ^*$, agent 2 violates PROP1 as removing a boundary item yields value at most $-\frac{1}{2} - \epsilon < -\frac{1}{2}$ for her. 
Is not hard to verify that allocation $\mathbf{B}$ with $B_1 =\{ e_1, e_2, e_3 \}$ and $ B  _2 = E \setminus B_ 1 $ is PROP1 and achieves the maximum utilitarian welfare $\UW(\mathbf{B}) = -1 - 4\epsilon$. Therefore, the price of PROP1 with respect to utilitarian welfare is bounded by 
$$
\PoF \geq \frac{-1 - 4\epsilon}{ - \frac{1}{2} - 2\epsilon} \rightarrow 2, \textnormal{ when } \epsilon \rightarrow 0,
$$ which completes the proof.
\end{proof}

%

\end{document}